\documentclass[10pt,letterpaper]{amsart}
\usepackage{amsmath,amssymb}
\usepackage{mathtools}
\usepackage{bbm}
\usepackage[all,cmtip]{xy}
\usepackage{stmaryrd} %
\usepackage{makecell}

\usepackage[colorinlistoftodos]{todonotes}

\usepackage{graphicx}

\usepackage{fonttable}

\usepackage{datetime}

\usepackage{bm}

\usepackage{wrapfig}

\makeindex \topmargin=0in \headheight=8pt \headsep=0.5in \textheight=8.4in %
\usepackage[style=alphabetic]{biblatex}
\usepackage{booktabs}
\bibliography{references}

\numberwithin{equation}{section}

\theoremstyle{plain}

\newtheorem{thma}{Theorem}

\newtheorem{assum}{Assumption}%
\newtheorem{property}{Property}%

\newtheorem{lemma}{Lemma}[section]
\newtheorem{thmas}[lemma]{Theorem}
\newtheorem{cor}[lemma]{Corollary}
\newtheorem{prop}[lemma]{Proposition}

\newtheorem{defin}[lemma]{Definition} %

\theoremstyle{remark}
\newtheorem{rem}{Remark}[section] %

\usepackage{xcolor}

\usepackage{eucal}
\usepackage{enumitem}

\usepackage{hyperref}
\hypersetup{
    colorlinks=true,
    linkcolor=blue,
    filecolor=magenta,      
    urlcolor=cyan,
    citecolor=blue,
}

\newcommand{\nc}{\newcommand}

\nc{\Aa}{{\CMcal{A}}}
\nc{\Bb}{{\CMcal{B}}}
\nc{\Cc}{{\CMcal{C}}}
\nc{\Dd}{{\CMcal{D}}}
\nc{\Ee}{{\CMcal{E}}}
\nc{\Ff}{{\CMcal{F}}}
\nc{\Gg}{{\CMcal{G}}}
\nc{\Hh}{{\CMcal{H}}}
\nc{\Ii}{{\CMcal{I}}}
\nc{\Jj}{{\CMcal{J}}}
\nc{\Kk}{{\CMcal{K}}}
\nc{\Ll}{{\CMcal{L}}}
\nc{\Mm}{{\CMcal{M}}}
\nc{\Nn}{{\CMcal{N}}}
\nc{\Oo}{{\CMcal{O}}}
\nc{\Pp}{{\CMcal{P}}}
\nc{\Qq}{{\CMcal{Q}}}
\nc{\Rr}{{\CMcal{R}}}
\nc{\Ss}{{\CMcal{S}}}
\nc{\Tt}{{\CMcal{T}}}
\nc{\Uu}{{\CMcal{U}}}
\nc{\Vv}{{\CMcal{V}}}
\nc{\Ww}{{\CMcal{W}}}
\nc{\Xx}{{\CMcal{X}}}
\nc{\Yy}{{\CMcal{Y}}}
\nc{\Zz}{{\CMcal{Z}}}

\nc{\mA}{{\mathrm{A}}}
\nc{\mB}{{\mathrm{B}}}
\nc{\mC}{{\mathrm{C}}}
\nc{\mD}{{\mathrm{D}}}
\nc{\mE}{{\mathrm{E}}}
\nc{\mF}{{\mathrm{F}}}
\nc{\mG}{{\mathrm{G}}}
\nc{\mH}{{\mathrm{H}}}
\nc{\mI}{{\mathrm{I}}}
\nc{\mJ}{{\mathrm{J}}}
\nc{\mK}{{\mathrm{K}}}
\nc{\mL}{{\mathrm{L}}}
\nc{\mM}{{\mathrm{M}}}
\nc{\mN}{{\mathrm{N}}}
\nc{\mO}{{\mathrm{O}}}
\nc{\mP}{{\mathrm{P}}}
\nc{\mQ}{{\mathrm{Q}}}
\nc{\mR}{{\mathrm{R}}}
\nc{\mS}{{\mathrm{S}}}
\nc{\mT}{{\mathrm{T}}}
\nc{\mU}{{\mathrm{U}}}
\nc{\mV}{{\mathrm{V}}}
\nc{\mW}{{\mathrm{W}}}
\nc{\mX}{{\mathrm{X}}}
\nc{\mY}{{\mathrm{Y}}}
\nc{\mZ}{{\mathrm{Z}}}

\nc{\BB}{{\mathbb{B}}}
\nc{\CC}{{\mathbb{C}}}
\nc{\DD}{{\mathbb{D}}}
\DeclareMathOperator{\EE}{{\mathbb{E}}}
\nc{\FF}{{\mathbb{F}}}
\nc{\GG}{{\mathbb{G}}}
\nc{\HH}{{\mathbb{H}}}
\nc{\II}{{\mathbb{I}}}
\nc{\JJ}{{\mathbb{J}}}
\nc{\KK}{{\mathbb{K}}}
\nc{\LL}{{\mathbb{L}}}
\nc{\MM}{{\mathbb{M}}}
\nc{\NN}{{\mathbb{N}}}
\nc{\OO}{{\mathbb{O}}}
\nc{\PP}{{\mathbb{P}}}
\nc{\QQ}{{\mathbb{Q}}}
\nc{\RR}{{\mathbb{R}}}

\nc{\TT}{{\mathbb{T}}}
\nc{\UU}{{\mathbb{U}}}
\nc{\VV}{{\mathbb{V}}}
\nc{\WW}{{\mathbb{W}}}
\nc{\XX}{{\mathbb{X}}}
\nc{\YY}{{\mathbb{Y}}}
\nc{\ZZ}{{\mathbb{Z}}}

\DeclareMathOperator{\Id}{\mathrm{Id}}

\nc{\Fun}{\mathrm{Fun}}

\def\dist{\mathrm{dist}}
\def\Area{\mathrm{Area}}
\def\diam{\mathrm{diam}}

\DeclareMathOperator{\supp}{\mathrm{supp}}

\let\Re\relax
\let\Im\relax
\DeclareMathOperator{\Re}{\mathrm{Re}}
\DeclareMathOperator{\Im}{\mathrm{Im}}

\nc\chr[2]{\begin{bmatrix}#1 \\ #2\end{bmatrix}}

\def\eps{\varepsilon}   
\def\vphi{\varphi}
\def\vtheta{\vartheta}
\def\cst{\mathrm{cst}}
\DeclareMathOperator{\osc}{\mathrm{osc}}

\let\div\relax
\DeclareMathOperator{\div}{\mathrm{div}}

\DeclareMathOperator{\Jac}{\mathrm{Jac}}

\def\loc{\mathrm{loc}}

\nc{\indic}{1\!\!1}
\DeclareMathOperator{\Var}{\mathrm{Var}}

\def\EL{\mathrm{EL}}

\def\smm{\smallsetminus}

\def\out{\mathrm{out}}

\def\diag{\mathrm{diag}}

\def\wtd{\widetilde}

\def\Lip{\mathrm{Lip}}

\def\ExpFat{\mathrm{ExpFat}}

\def\spl{\mathrm{spl}}

\def\Wws{W^\circ_\spl}

\def\Bbs{B^\bullet_\spl}

\def\genbox#1#2#3#4#5#6{%
    \leavevmode\raise#4bp\hbox to#5bp{\vrule height#5bp depth0bp width0bp
    \pdfliteral{q .5 w \csname #2COLOR\endcsname\space RG
                       \csname #3PDF\endcsname{#5}{#6} S Q
             \ifx1#1 q \csname #2COLOR\endcsname\space rg 
                       \csname #3PDF\endcsname{#5}{#6} f Q\fi}\hss}}

\begin{document}

\title[Harmonic functions on Tutte embeddings]
{Harmonic functions on Tutte embeddings\\ and linearized Monge--Amp\`ere equation}

\author[M. Basok]{Mikhail Basok$^\textsc{a}$}
\author[D. Chelkak]{Dmitry Chelkak$^\textsc{b}$}
\author[B. Laslier]{Beno\^it Laslier$^\textsc{c}$}
\author[M. Russkikh]{Marianna Russkikh$^\textsc{d}$}

\thanks{\textsc{${}^\mathrm{A}$ University of Helsinki, Department of Mathematics and Statistics, 00500 Helsinki, Finland}}
\thanks{\textsc{${}^\mathrm{B}$ University of Michigan, Department of Mathematics, Ann Arbor, MI 48109-1043, USA}}
\thanks{\textsc{${}^\mathrm{C}$ Universit\'e Paris Cit\'e and Sorbonne Universit\'e, CNRS, Laboratoire de
Probabilit\'es, Statistique et Mod\'elisation, F-75013 Paris, France and
D\'epartement de Math\'ematiques et Applications, \'Ecole normale sup\'erieure, Universit\'e PSL, CNRS, 75005 Paris, France}}
\thanks{\textsc{${}^\mathrm{D}$ University of Notre Dame, Department of Mathematics, 255 Hurley Bldg, Notre Dame, IN
46556, United States of America}}
\thanks{\textit{E-mail:} \texttt{m.k.basok@gmail.com}, \texttt{dchelkak@umich.edu}, \texttt{laslier@lpsm.paris}, \texttt{mrusskik@nd.edu}}

\keywords{discrete harmonic functions, discrete holomorphic functions, linearized Monge--Amp\`ere equation, conjugate Beltrami equation, conformal invariance, t-embeddings and t-surfaces}

\subjclass[2020]{Primary 60G50; Secondary 30G25, 35J96}

\maketitle

\begin{abstract} 
  We prove convergence of solutions of Dirichlet problems and Green's functions on Tutte harmonic embeddings to those of the linearized Monge--Amp\`ere equation~$\Ll_\varphi h=0$. More precisely, we assume that piecewise linear Maxwell--Cremona potentials associated with the embeddings converge to a continuous potential $\varphi$ and the only assumption that we use is the uniform convexity of~$\varphi$ or, equivalently, the uniform ellipticity of the operator~$\Ll_\varphi$. Even if~$\varphi$ is quadratic, this setup significantly generalizes known results for discrete harmonic functions on orthodiagonal tilings. Motivated by potential applications to the analysis of 2d lattice models on irregular graphs, we also study the situation in which the limits are harmonic in a different complex structure.
\end{abstract}

\renewcommand\vphi{\varphi}

\newcommand\vpsi{\psi}
\newcommand\vPhi{\Phi}
\newcommand\vPsi{\Psi}

\section{Introduction and main results}
\label{sec:Introduction}

\subsection{Motivation and informal overview of the main results} 
Convergence of discrete harmonic functions in the small mesh size limit is a classical topic dating back to the first works on numerical methods~\cite{Lusternik1926,courant-friedrichs-lewy-1928}. In two dimensions, this question is intrinsically related to the notion of discrete holomorphic functions and their convergence, with important contributions made in the middle of the last century by Ferrand~\cite[Chapter~V]{lelong-ferrand-book} who greatly developed this theory on the square grid and Duffin~\cite{duffin-rhombic} who, in particular, suggested a generalization to the so-called rhombic tilings of the complex plane. From the combinatorial perspective, two related classical setups for which the notions of discrete holomorphic and harmonic functions make perfect sense are Brooks--Smith--Stone--Tutte's square (or, more generally, rectangular) tilings~\cite{brooks-smith-stone-tutte-1940} and Tutte's barycentric, or harmonic, embeddings~\cite{tutte1963draw}. One of the main sources of the modern interest to discrete complex analysis is the statistical physics in two dimensions; e.g., see~\cite{smirnov-icm2010}. In particular, the framework of rhombic lattices was important for the study of scaling limits of free fermionic models with Baxter's Z-invariant weights; e.g., see~\cite{mercat-cmp2001,kenyon-isorad,chelkak-smirnov-universality,boutillier-detiliere-rashel-massiveLaplacian}. 

Recent advances of the probabilistic approach to the Liouville Quantum Gravity (LQG) in~2d (e.g., see~\cite{gwynne-holden-sun-matingLQG-survey} and references therein) motivate a development of discrete complex analysis techniques for possibly very irregular planar grids that appear as embeddings of given (random) planar graphs into the plane. In particular, applications of Tutte embeddings and square tilings in the LQG context recently appeared in~\cite{gwynne-miller-sheffield-Tutte-mated-CRT,bertacco-gwynne-sheffield-smith-embeddings}. Going in this direction, the last three authors of this paper suggested a general framework of the so-called t-embeddings of graphs equipped with a bipartite dimer model~\cite{CLR1}, also known as Coulomb gauges~\cite{KLRR}, which simultaneously generalizes Tutte embeddings and the so-called s-embeddings of planar graphs carrying the Ising model. Shortly afterwards, it was understood that t-embeddings should be viewed as piecewise linear surfaces in the Minkowski space~$\RR^{2,2}$, called \emph{t-surfaces} below, which provide a correct description of the complex structure arising in the small mesh size limit~\cite{chelkak-ramassamy-aztec,chelkak-semb,CLR2,berggren-nicoletti-russkikh-doubleaztec}.

In this paper we apply the aforementioned framework to study discrete harmonic functions on Tutte embeddings. To the best of our knowledge, recent literature on the subject was mostly focused on the very special case -- so-called \emph{orthodiagonal tilings} -- which appear when the Tutte embeddings of a graph and of its dual form non-self-intersecting quads. In this setup, the convergence to usual harmonic functions (that is, solutions of the Laplace equation $\Delta h = 0$) under no regularity assumptions on the tilings was proven in~\cite{gurel-gurevich-jerison-nachmias}. A similar result in higher dimensions was obtained in~\cite{bou2024random} under a mild regularity assumption; see also references therein and~\cite{binder2024orthodiagonal}. 
While it is sometimes argued that orthodiagonal tilings may be the right setup to study the universality of lattice models via discrete complex analysis tools, the passage from Tutte embeddings to orthodiagonal tilings is rather restrictive: each weighted planar graph admits the former but not necessarily the latter. 
This calls for a theory that allows to prove similar convergence results on generic Tutte embeddings.

Using the classical Maxwell--Cremona correspondence, with each Tutte embedding~$\Hh$ one can associate a piecewise linear convex function~$\vPhi$; we call this function a \emph{Maxwell--Cremona potential;} see Section~\ref{subsec:intro-gamma-phi-def} below for more details. 
Given a sequence $\Hh^{(n)}$ of Tutte embeddings such that the corresponding potentials $\vPhi^{(n)}$ converge to a uniformly convex (see~\eqref{eq:unif_conv}) function $\varphi$ as $n\to\infty$, we prove that discrete harmonic functions on $\Hh^{(n)}$ converge to solutions of the \emph{linearized Monge--Amp\`ere equation} with potential $\vphi$; see Theorem~\ref{thma:harmonic_functions_convergence} and Theorem~\ref{thma:Green_fct_convergence} for precise statements. Let us emphasize that the uniform convexity of $\vphi$ is the only assumption that we use besides the convergence of $\vPhi^{(n)}$ to $\varphi$; the latter is clearly necessary to describe the limit. In particular, similarly to~\cite{gurel-gurevich-jerison-nachmias}, we do not impose any additional assumptions on the local structure of the embeddings $\Hh^{(n)}$. 

In the situation when Tutte embeddings~$\Hh^{(n)}$ appear from orthodiagonal tilings, one can easily see that the potentials~$\vPhi^{(n)}$ approximate the quadratic potential~$\vphi(w)=\tfrac12|w|^2$ as $n\to\infty$. Let us emphasize that the converse is not true: $\Hh^{(n)}$ do not necessarily come from orthodiagonal tilings if 
the limits of discrete harmonic functions on $\Hh^{(n)}$ are described by the usual Laplacian; see Remark~\ref{rem:intro_potential_isorad_ortho} below. This makes our results more general than those of~\cite{gurel-gurevich-jerison-nachmias} even in this situation. In general, the small mesh size limits of discrete harmonic functions may not be harmonic in the usual sense even after a non-trivial change of coordinates. 
We provide several equivalent conditions characterizing the existence of such a coordinate change in Theorem~\ref{thma:when_harmonic}.

The uniform convexity assumption~\eqref{eq:unif_conv} is necessary for the linearized Monge--Amp\`ere equation that appears in the limit to be uniformly elliptic. Moreover, in Theorem~\ref{thma:LIP<=>RW} we push this to the discrete level \emph{before} passing to the limit. Namely, we prove that the Maxwell--Cremona potential~$\vPhi$ of a given embedding~$\Hh$ is uniformly convex above a certain scale~$\delta$ if and only if the random walk on $\Hh$ satisfies standard uniform crossing estimates above scale $\delta$ and its natural invariant measure is uniformly comparable to the Euclidean one above the same scale; see property~\ref{prty:RW} below. 

Also, it is worth noting that the uniform convexity of~$\vphi$ is necessary to define harmonic conjugates of solutions of the linearized Monge--Amp\`ere equation~\eqref{eq:def_of_Ll} that are usually referred as $A$-harmonic conjugates in the literature; e.g., see~\cite[Chapter~16]{AstalaBook}. This notion has a discrete counterpart: classically, each discrete harmonic function on a Tutte embedding has a discrete harmonic conjugate, which is defined on the dual embedding. In fact, in Theorem~\ref{thma:harmonic_functions_convergence} we prove the simultaneous convergence of discrete harmonic functions and their harmonic conjugates to a solution of a linearized Monge--Amp\`ere equation and its $A$-conjugate, respectively. In particular, the uniform convexity of $\vphi$ is the minimal assumption to be imposed if one is interested in such a simultaneous convergence of discrete harmonic functions and their harmonic conjugates.

\subsection{Maxwell--Cremona correspondence}
\label{subsec:intro-gamma-phi-def}

Let $\Gamma$ be a planar graph drawn in the complex plane so that its edges are non-intersecting straight segments. We denote by $v$ the vertices of this graph and by $\Hh(v)\in\CC$ the coordinates of these vertices in the complex plane. Each edge $(vv')$ has a pre-assigned weight, or \emph{conductance,} $c_{vv'}>0$. The embedding~$\Hh$ is called \emph{Tutte,} \emph{barycentric,} or \emph{harmonic} if the identity $\sum_{v'\sim v} c_{vv'}(\Hh(v')-\Hh(v))=0$ holds for each inner vertex of the graph.

Given a harmonic embedding $\Hh$, 
one can define a function $\Hh^\ast$ on the dual graph $\Gamma^\ast$ by setting
\begin{equation}
\label{eq:intro_H*def}
\Hh^\ast(v^\ast_2)-\Hh^\ast(v^\ast_1)=ic_{v_1v_2}(\Hh(v_2)-\Hh(v_1))
\end{equation}
if $(v^\ast_1v^\ast_2)$ is the dual edge to $(v_1v_2)$ and $v^\ast_1$ is to the right from $(v_1v_2)$. The harmonicity of $\Hh$ implies that the sum of these increments around each inner vertex of~$\Gamma$ vanishes and hence $\Hh^\ast$ -- %
the \emph{dual} Tutte embedding -- is well defined up to a global additive constant. Let us emphasize that this dual embedding can have overlaps if the boundary of $\Hh(\Gamma)$ is not convex; see Section~\ref{sec:T-embedding of corner graph} for more details.

A crucial role for our results is played by the \emph{Maxwell--Cremona lift} of $\Hh$; see~\cite{pseudo-triang-survey, schulze2017rigidity}. 
Denote by $U$ the union of faces of $\Gamma$ and define a \emph{piecewise constant map} $\vPsi:U\to\CC$ by setting
\begin{equation}
\label{eq:intro-Psi-def}
\vPsi(w):=\Hh^\ast(v^\ast)\quad \text{\makecell[l]{if $w\in\CC$ belongs to the face of $\Gamma$ \\ corresponding to the dual vertex $v^\ast$.}}
\end{equation}
(If $w\in\CC$ is a vertex of $\Gamma$ or lies on an edge, we choose one of the incident faces arbitrarily.) We then define a piecewise linear \emph{continuous} function $\vPhi:U\to\RR$ by setting
\[
\vPsi=2\partial_{\bar{w}}\vPhi=\partial_x\vPhi+i\partial_y\vPhi\quad \text{for all $w=x+iy$ lying on faces of $\Gamma$}\,.
\]
Equivalently, we require that the real-valued function~$\Phi$ is linear on faces of $\Hh(\Gamma)$ and
\begin{equation}%
  \label{eq:intro-Phi-def}
  \vPhi(\Hh(v_2))-\vPhi(\Hh(v_1))=\Re\big[\,\overline{\Hh^\ast(v^\ast_1)}(\Hh(v_2)-\Hh(v_1))\big]%
  =\Re\big[\,\overline{\Hh^\ast(v^\ast_2)}(\Hh(v_2)-\Hh(v_1))\big]
\end{equation}%
for each pair of incident vertices $v_1,v_2$, where $v_1^\ast,v_2^\ast$ are the two faces of $\Gamma$ incident to the edge $(v_1v_2)$. 
A classic observation due to Maxwell~\cite{maxwell1864xlv, maxwell1870reciprocal, cremona1872figure} asserts that the function $\vPhi$ is well defined. Also, it is easy to see that $\vPhi$ is defined by $\Hh$ uniquely up to a global linear term.

Vice versa, given a piecewise linear function $\Phi$ on a domain $U\subset\CC$ one can reconstruct the embedding $\Hh$ of a graph $\Gamma$ 
and also the edge weights $c_{vv'}$, known as \emph{self-stresses}, via~\eqref{eq:intro_H*def} and~\eqref{eq:intro-Psi-def}.
A simple computation shows that these weights are positive if and only if the piecewise linear function~$\vPhi$ is convex on each convex subset of~$U$.

\begin{rem}\label{rem:intro_potential_isorad_ortho} (i) The most classical setup for discrete complex analysis beyond the square grid is that of rhombic lattices or isoradial graphs introduced by Duffin in~\cite{duffin-rhombic}; see also~\cite{chelkak-smirnov-isorad}. In this case, it is easy to deduce from~\eqref{eq:intro_H*def} that the restriction of $\Phi$ onto the vertices of~$\Gamma$ equals $\vPhi(\Hh(v))=\frac12|\Hh(v)|^2$. 

\noindent (ii) In a more general setup of orthodiagonal tilings~\cite{gurel-gurevich-jerison-nachmias,binder2024orthodiagonal,bou2024random}, one has $\Psi(w)=w+{O(\delta)}$ and hence $\vPhi(w)=\frac12|w|^2+{O(\delta|w|)}$, where $\delta\to 0$ denotes the size of the largest tile. As for the converse, let us emphasize that the assumption $\Psi(w)=w+{O(\delta)}$ does \emph{not} put the analysis back to the setup of~\cite{gurel-gurevich-jerison-nachmias,binder2024orthodiagonal,bou2024random} 
even if $\delta$ is comparable to the lengths of all the edges of $\Gamma$: the reason is that
the quads $\Hh(v_1)\Hh^\ast(v^\ast_1)\Hh(v_2)\Hh^\ast(v^\ast_2)$ may self-intersect. 
See also Remark~\ref{rem:intro-convergence-of-processes} below.
\end{rem}

Below we often identify vertices $v$ of the graph~$\Gamma$ and their positions $\Hh(v)$ if no confusion arises. 

\subsection{Uniform convexity of $\vphi$ and auxiliary scales $\delta^{(n)}$}
\label{subsec:scales-delta_n}

Let $U\subset \CC$ be a domain and \mbox{$\vphi:U\to \RR$} be a function that is convex on each convex subset of $U$. We say that $\vphi$ is \emph{uniformly convex} on $U$ if there exists 
a $\lambda>0$ such that %
\begin{equation}
  \label{eq:unif_conv}
  \lambda|w_1 - w_2|^2\ \leq\ \vphi(w_1) - 2\vphi\left( \tfrac{w_1+w_2}2 \right) + \vphi(w_2)\ \leq\ \lambda^{-1}|w_1-w_2|^2
\end{equation}
for each segment $[w_1;w_2]\subset U$. In our main convergence results we consider a sequence of Tutte embeddings $\Gamma^{(n)}$ whose Maxwell--Cremona potentials $\vPhi^{(n)}$ are all defined on $U$ and converge to a uniformly convex potential $\vphi$ as $n\to\infty$. Since all functions $\Phi^{(n)}$ are convex, one can always replace $U$ with a slightly smaller domain and assume that this convergence is uniform. %
In other words, without true loss of generality we can assume that
\begin{equation}
  \label{eq:def_of_deltan}
  \delta^{(n)} = \|\Phi_n - \vphi\|^{1/2}_{\mL^\infty(U)}\to 0\ \ \text{as}\ n\to\infty.
\end{equation}

The uniform convexity of $\vphi$ implies that the Maxwell--Cremona potentials $\Phi^{(n)}$ satisfy the same property~\eqref{eq:unif_conv} with a slightly smaller $\lambda>0$ if we additionally require that $|w_2-w_1|\geq C\delta^{(n)}$ for a big enough constant $C>0$. Thus, instead of a sequence of Tutte embeddings~$\Gamma^{(n)}$, $n\to\infty$, whose Maxwell--Cremona potentials converge to a uniformly convex potential $\varphi$ we can consider a sequence of Tutte embeddings $\Gamma_\delta$, $\delta\to 0$, that satisfy the following property:

\renewcommand{\theproperty}{(CONV)}
\begin{property}
\label{prty:CONV}\label{assum:conv} Given constants $\lambda,C,\delta>0$, we say that a Tutte embedding $\Gamma_\delta$ with Maxwell-Cremona potential~$\vPhi_\delta$ has property~\ref{prty:CONV} on a domain $U\subset \CC$ if the inequalities 
\begin{equation}
\label{eq:Conv}
\lambda\,|w_2-w_1|^2\ \le\ \vPhi_\delta(w_2)-2\vPhi_\delta(\tfrac12(w_1\!+\!w_2))+\vPhi_\delta(w_1)\ \le\ \lambda^{-1} |w_2-w_1|^2 
\end{equation}
hold for all straight segments $[w_1;w_2]\subset U$ such that $|w_2-w_1|\ge C\delta$.

Further, we say that a sequence of Tutte embeddings~$\Gamma_\delta$, $\delta\to 0$, satisfies property~\ref{prty:CONV} on a domain $U\subset\CC$ if this property holds for each $\Gamma_\delta$ with constants $\lambda,C>0$ that do not depend on $\delta$.
\end{property}

For the sake of notational simplicity, from now onwards we stick to the notation $\Gamma_\delta$ for Tutte embeddings instead of~$\Gamma^{(n)}$ and label all related objects using the auxiliary scales $\delta=\delta^{(n)}$. Note again that we do not lose generality as one can always define these scales by~\eqref{eq:def_of_deltan} and then find constants $C,\lambda>0$ depending on the potential $\varphi$ only so that~\eqref{eq:Conv} holds. Let us additionally emphasize that the scale $\delta$ can be \emph{much bigger} than the maximal diameter of faces of $\Gamma_\delta$. Note also that property~\ref{assum:conv} does not imply any bound on the degrees of vertices of $\Gamma_\delta$ or the angles of its faces.

In terms of piecewise constant mappings $\vPsi_\delta=2\partial_{\bar{w}}\vPhi_\delta:U\to\CC$ associated with $\Gamma_\delta$ by~\eqref{eq:intro-Psi-def}, property~\ref{prty:CONV} can be reformulated (in the bulk of $U$) as follows: 
\renewcommand{\theproperty}{(LIP)}
\begin{property}
  \label{prty:LIP}\label{assum:main}
  Given constants $\varkappa,C,\delta>0$, we say that a Tutte embedding $\Gamma_\delta$ with the map $\Psi_\delta$ defined by~\eqref{eq:intro-Psi-def} has property~\ref{prty:LIP} on a domain $U\subset\CC$ if 
  \begin{equation}
        \label{eq:Lip} 
        \Re \frac{\vPsi_\delta(w_2) - \vPsi_\delta(w_1)}{w_2-w_1}\,\geq\,\varkappa \quad \text{and}\quad \left|\frac{\vPsi_\delta(w_2) - \vPsi_\delta(w_1)}{w_2-w_1}\right|\,\leq\,\varkappa^{-1}
  \end{equation}
  for all straight segments $[w_1;w_2]\subset U$ such that $|w_2-w_1|\ge C\delta$.
  
  Further, we say that a sequence of Tutte embeddings~$\Gamma_\delta$, $\delta\to 0$, satisfies property~\ref{prty:LIP} on a domain $U\subset\CC$ if this property holds for each $\Gamma_\delta$ with constants $\varkappa,C>0$ that do not depend on $\delta$.
\end{property}
\begin{rem}\label{rem:intro-conv<=>lip+k-lip}
It is not hard to see that property~\ref{prty:CONV} for potentials $\Phi_\delta$ on $U$ implies property~\ref{prty:LIP} for their gradients~$\Psi_\delta$ on each subdomain $U'\Subset U$ and vice versa; we refer the reader to Lemma~\ref{lemma:Lip_equiv_conv} for the proof of this elementary fact. (It is sufficient to assume that $U'$ lies in an $O(\delta)$-interior of $U$.) More importantly, property~\ref{prty:LIP} is also equivalent to the $\kappa$-Lipschitzness, $\kappa=\kappa(\varkappa)<1$, of the so-called origami map in the terminology of~\cite{CLR1}, which makes the techniques developed in~\cite{CLR1} applicable to the setup of this paper; see Section~\ref{subsec:equivalence_of_Lip} below. %
\end{rem}

\subsection{Linearized Monge--Amp\`ere equation}
\label{subsec:LMA}

Let $\vphi:U\to \RR$ be a uniformly convex function. It follows that $\vphi\in\mC^{1,1}(U)$, that is, the gradient of $\vphi$ is a Lipshitz function (e.g., see Lemma~\ref{lemma:Lip_equiv_conv} below for a similar statement).
Denote
\begin{equation}
  \label{eq:def_of_Ll} 
  \Ll_\vphi h = -\vphi_{yy}h_{xx}+2\vphi_{xy}h_{xy}-\vphi_{xx}h_{yy} = -\div (A_\vphi\nabla h),\qquad 
  A_\vphi\ =\ \begin{pmatrix} \vphi_{yy} & -\vphi_{xy}\\ -\vphi_{xy} & \vphi_{xx}\end{pmatrix}.
\end{equation}
The matrix $A_\vphi$ is defined almost everywhere and is uniformly elliptic, that is, $\lambda\le A_\vphi\le \lambda^{-1}$, due to Rademacher's theorem and the assumption~\eqref{eq:unif_conv}.
The equation $\Ll_\vphi h=0$ is known under the name \emph{linearized Monge--Amp\`ere equation}~\cite{caffarelli-gutierrez-1997-lin-MA}. Even though the coefficients of the matrix~$A_\vphi$ are merely measurable, the fact that the operator~$\Ll_\vphi$ can be written in the divergence form
allows one to speak about \emph{weak solutions,} that is, functions $h\in W^{1,2}_\loc(U)$ such that $\int_U (A_\vphi \nabla h\cdot \nabla\phi)(w)d^2w=0$ for all smooth test functions~$\phi\in \mC_0^\infty(U)$.  

Another feature coming from the divergence form of $\Ll_\vphi$ is the notion of $A_\vphi$-harmonic conjugates.  
For each weak solution of the equation~$\Ll_\vphi h=0$ one can (locally) find a function $h^\ast\in W_\loc^{1,2}(\Omega)$ such that
\[
  \nabla h^\ast = \ast A_\vphi \nabla h,\quad \text{where}\quad \ast = \left(\begin{smallmatrix} 0 & -1 \\ 1 & 0\end{smallmatrix}\right)
\]
is the Hodge star operator in $\RR^2$. The function $h^\ast$, defined up to an additive constant, is called an \emph{$A_\vphi$-harmonic conjugate} of~$h$; see~\cite[Chapter~16]{AstalaBook} and Section~\ref{sec:Elliptic equations and closed 1-forms in 2D} below for more details. Note that $h^\ast$ is a weak solution of {the `dual' (\emph{not} in the sense of a scalar product) equation}
\begin{equation}
  \label{eq:def_of_Lltimes}
\Ll^\times_\vphi h^\ast = \div (\,\ast\,A_\vphi^{-1}\!\ast\nabla h^\ast) = 0.
\end{equation}
Standard results from the theory of elliptic equations guarantee that weak solutions~$h,h^\ast$ are continuous in~$\Omega$; see~\cite{kenig2000new} and references therein. Moreover, if $\Omega\Subset U$ is a compactly supported subdomain, then the Dirichlet problem~$\Ll_\vphi h=0$ in~$\Omega$, $h\vert_{\partial\Omega}=g$, has a unique solution $h\in W^{1,2}_\loc(\Omega)\cap \mC(\overline{\Omega})$ for each continuous function~$g:\partial\Omega\to\RR$~\cite{littman-stampacchia-weinberger-1963-regular-points, kenig2000new}.

\begin{rem}
  \label{rem:Ll_Phi}
  Although the potential $\vPhi_\delta$ corresponding to a given Tutte embedding~$\Gamma_\delta$ is not twice differentiable, one can still consider the quadratic form of the operator $\Ll_{\vPhi_\delta}$ on smooth compactly supported functions defined on edges of $\Gamma_\delta$. This quadratic form gives rise to a Brownian motion on the cable system (metric graph) of $\Gamma_\delta$, which explains why linearized Monge--Amp\`ere equations provide the right language for studying the convergence of discrete harmonic functions on $\Gamma_\delta$ as $\delta\to 0$.
\end{rem}

\subsection{Main convergence results}
\label{subsec:main_convergence}
Let $\Gamma_\delta$ be a Tutte embeddings that covers a domain~$U\subset\CC$. Fix a subdomain $\Omega\Subset U$ and let $\Omega_\delta$ be the subgraph of $\Gamma_\delta$ spanned by those vertices of $\Gamma_\delta$ that lie inside~$\Omega$. (If this intersection has more than one connected component, we keep only the one that covers the bulk of $\Omega$.) Let $\overline{\Omega}_\delta$ be the union of $\Omega_\delta$ and the set of all neighbors of these vertices, and denote $\partial\Omega_\delta=\overline{\Omega}_\delta\smallsetminus\Omega_\delta$. Given a function $H_\delta:\overline{\Omega}_\delta\to \RR$ on vertices of $\overline{\Omega}_\delta$ and $v\in\Omega_\delta$, we define
\begin{equation}
  \label{eq:def_of_Lldelta}
  [\Ll_\delta H_\delta](v) = (\mu_\delta(v))^{-1}\sum\nolimits_{v'\sim v} c_{vv'} (H_\delta(v) - H_\delta(v'))
\end{equation}
where the normalizing factors are given by
\begin{equation}
\label{eq:intro-mu-def}
\mu_\delta(v):=\sum\nolimits_{v'\sim v} c_{v'v}|v-v'|^2.
\end{equation}
A function $H_\delta$ is called \emph{discrete harmonic} in $\Omega_\delta$ if $[\Ll_\delta H_\delta](v)= 0$ for all $v\in\Omega_\delta$. If $\Omega$ is simply connected, then with each discrete harmonic function~$H_\delta$ in $\Omega_\delta$ one can associate its discrete harmonic conjugate~$H^\ast_\delta$ defined on faces of $\Gamma_\delta$ that are contained in $\Omega$ so that
\begin{equation}
\label{eq:intro_harm_conj}
H^\ast(v^\ast_2)-H^\ast(v^\ast_1)=c_{v_1v_2}(H(v_2)-H(v_1)),
\end{equation}
where $(v^\ast_1v^\ast_2)$ is the dual edge to $(v_1v_2)$ and $v^\ast_1$ lies to the right from $(v_1v_2)$. 
The function $H^\ast_\delta$ is discrete harmonic on the portion of the dual embedding $\Gamma^*_\delta$ on which it is defined. 

Our first convergence theorem deals with solutions of discrete Dirichlet problems. Recall that, given a sequence of Tutte embeddings $\Gamma^{(n)}$ whose Maxwell--Cremona potentials $\vPhi^{(n)}$ converge to a uniformly convex function $\varphi$ as $n\to\infty$, we re-label them as $\Gamma_\delta$ using auxiliary scales $\delta=\delta^{(n)}\to 0$ defined by~\eqref{eq:def_of_deltan}. As discussed in Section~\ref{subsec:scales-delta_n}, the unifrom convexity of~$\varphi$ implies that, after this re-labeling, the embeddings $\Gamma_\delta$ satisfy property~\ref{prty:CONV}. Vice versa, if $\Gamma_\delta$ satisfy property~\ref{prty:CONV}, then each limit of the corresponding potentials $\vPhi_\delta$ as $\delta\to 0$ is automatically uniformly convex.

\begin{thma}
  \label{thma:harmonic_functions_convergence}
  Let $\Gamma_\delta$ be a sequence of Tutte embeddings that satisfy property~\ref{prty:CONV} on a domain $U\subset \CC$ such that their Maxwell--Cremona potentials $\Phi_\delta$ converge to a function $\vphi:U\to \RR$ as $\delta\to 0$. Let $\Omega\Subset U$ and $g$ be a continuous function defined in a neighborhood of $\partial \Omega$. Denote by $h$ %
  the solution of the Dirichlet problem $\Ll_\vphi h = 0$, $h\vert_{\partial \Omega} = g$, and let $H_\delta: \overline{\Omega}_\delta\to \RR$ be the solution of the same Dirichlet problem for~$\Ll_{\delta}$, i.e., the unique discrete harmonic function in~$\Omega_\delta$ such that $H_\delta(v)=g(v)$ for all~$v\in\partial\Omega_\delta$. Then, the functions~$H_\delta$ converge to $h$ uniformly in $\overline \Omega$ as $\delta\to 0$. 
  
  Moreover, if $\Omega$ is simply connected, then discrete harmonic conjugates~$H_\delta^\ast$ can be chosen so that they converge to an $A_\vphi$-harmonic conjugate $h^\ast$ {of $h$} uniformly on compact subsets of $\Omega$.
\end{thma}

We also have a similar convergence result for Green's functions. Namely, let~$G_{\Omega_\delta}(\cdot,v_0):\overline{\Omega}_\delta\to\RR$ be 
the discrete Green function in~$\Omega_\delta$, i.e., 
the unique function that is discrete harmonic in $\Omega_\delta\smm\{v_0\}$, has zero boundary conditions on~$\partial\Omega_\delta$, and such that $\mu_\delta(v_0)[\Ll_\delta G_{\Omega_\delta}(\cdot,v_0)](v_0)=1$. In probabilistic terms, $\mu_\delta(v_0)G_{\Omega_\delta}(v,v_0)$ is the expected time spent at~$v_0$ by the continuous time random walk~$X_t$ on~$\Gamma_\delta$ started at $v$ and parameterized so that $|X_t|^2-t$ is a local martingale before it exists $\Omega_\delta$. Further, let $G_\Omega(\cdot,w_0)$ be the Green function of the elliptic operator $\Ll_\vphi$, that is, the unique positive function that belongs to the space $W^{1,2}_\loc(\Omega\smm\{ w_0 \})\cap \mC(\overline\Omega \smm\{ w_0 \})$, vanishes on $\partial \Omega$, and satisfies the identity
\[
\textstyle \int_\Omega A_\vphi \nabla_w G_\Omega(w,w_0)\cdot \nabla \phi(w)\,d^2w = \phi(w_0)
\]
for all smooth compactly supported test functions~$\phi\in \mC_0^\infty(\Omega)$.

\begin{thma}
  \label{thma:Green_fct_convergence}
  In the same setup as above, discrete Green functions~$G_{\Omega_\delta}(\,\cdot\,,\,\cdot\,)$ converge, as $\delta\to 0$, to the Green function $G_\Omega(\,\cdot\,,\,\cdot\,)$, uniformly on each compact $K\Subset (\overline\Omega \times \overline\Omega) \smm\diag$.
\end{thma}

We prove Theorem~\ref{thma:harmonic_functions_convergence} and Theorem~\ref{thma:Green_fct_convergence} in Section~\ref{sec:Convergence of harmonic functions in C0 topology}: first under an additional smoothness assumption~$\vphi\in\mC^3$ in Section~\ref{subsec:proof of Theorem12 vphi smooth} and then in the general case~$\vphi\in\mC^{1,1}$ in Section~\ref{subsec:Extension to C1,1}. The principle difference between the two cases is a regularity lemma asserting that, under the assumption $\vphi\in\mC^3$, every `very weak' solution of the equation~$\Ll_\vphi h=0$, i.e., a function~$h$ such that $\int_\Omega h(w)[\Ll_\vphi \phi](w) d^2w=0$ for all $\phi\in\mC^\infty_0(\Omega)$, is automatically a weak solution (see Lemma~\ref{lemma:very_weak=weak}). In the general case $\vphi\in\mC^{1,1}$ such a statement requires the existence of the gradient of $h$ to be known a priori.

To overcome this difficulty we choose to rely upon the a priori regularity theory developed in~\cite{CLR1} and an idea somewhat similar to the one used in~\cite{mahfouf-crossing-estimates} in the Ising model context. This allows us to show that each discrete harmonic function on~$\Gamma_\delta$ can be locally approximated by discrete harmonic functions on the same graph with bounded gradients, and hence to control the gradients of limits of discrete harmonic functions. Moreover, as a byproduct of this approach we show that the convergence in Theorem~\ref{thma:harmonic_functions_convergence} and Theorem~\ref{thma:Green_fct_convergence} holds also for the `gradients defined on a mesoscopic scale'; see Remark~\ref{rem:Hdelta_is_mesodiff}. While we do not know whether or not the actual gradients of discrete harmonic functions can be uniformly bounded in the general case, we note that the convergence 
of~$\nabla H_\delta$ and $\nabla G_{\Omega_\delta}(\cdot,v_0)$ holds provided that the graphs $(\Gamma_\delta)_{\delta\to 0}$ satisfy a very mild additional regularity assumption~\ref{assum:Exp-Fat}. We refer the reader to Section~\ref{subsec:C1 under ExpFat} and Section~\ref{subsec:Extension to C1,1} for details.

\begin{rem}
\label{rem:intro-convergence-of-processes} 
In the special setup of orthodiagonal tilings, Theorem~\ref{thma:harmonic_functions_convergence} has been obtained in~\cite{gurel-gurevich-jerison-nachmias}; see also references therein and~\cite[Theorem~B]{bou2024random} for an alternative proof. 
Moreover, in~\cite[Theorem~A]{bou2024random} the authors show that in this case \cite[Theorem~B]{bou2024random} implies the convergence of trajectories of the random walks on~$\Gamma_\delta$ to trajectories of the Brownian motion, up to a possible time-reparametrization. We believe that Theorem~\ref{thma:harmonic_functions_convergence} supplemented by Theorem~\ref{thma:Green_fct_convergence} is sufficient to prove the convergence of continuous time random walks on~$\Gamma_\delta$ to the driftless diffusion defined by %
$\Ll_\vphi$ \emph{together} with the time-parametrizations. However, we do not explore this question in our paper.
\end{rem}

\subsection{Ellipticity of random walks on Tutte embeddings}
\label{subsec:CONV-and-RW}

The uniform ellipticity of the operator $\Ll_\vphi$ given by~\eqref{eq:def_of_Ll} can be rephrased in probabilistic terms as follows: the diffusion process generated by~$\Ll_\varphi$ must satisfy uniform crossing estimates and propagate with the speed that is uniformly comparable to~$1$. Our next result asserts that a similar equivalence holds for the random walk on a Tutte embedding: $\Gamma_\delta$ satisfies property~\ref{assum:conv} if and only if the properly normalized random walk on $\Gamma_\delta$ is `uniformly elliptic' \emph{above} the scale $\delta$. Let us emphasize that this equivalence holds for each $\Gamma_\delta$ and does \emph{not} involve taking any limit. To formulate it precisely we need some notation.

Given a Tutte embedding $\Gamma_\delta$, we define the continuous time random walk~$X_t$ on its vertices so that the jump rates from a vertex $v$ to its neighbors~$v'$ are proportional to $c_{vv'}$, parametrized so that $|X_t|^2-t$ is a local martingale. (Note also that $X_t$ is a $\CC$-valued local martingale since the embedding is harmonic.) The normalization factors $\mu_\delta$ defined in~\eqref{eq:intro-mu-def}
give us a standard invariant measure for~$X_t$ if one excludes the effect of the boundary of $\Gamma_\delta$ from the consideration. %
\renewcommand{\theproperty}{(RW)} 
\begin{property} Given constants $c,C,\delta>0$, we say that a Tutte embedding $\Gamma_\delta$ has property~\ref{prty:RW} on a domain $U\subset\CC$ covered by $\Gamma_\delta$ if the following holds:
\label{prty:RW}\label{assum:RW} 
\begin{enumerate}[label=(\alph{enumi})]
  \item \label{prop:S} The random walk on $\Gamma_\delta$ satisfies usual crossing estimates starting from the scale $\delta$:
  for each vertex $v$ of $\Gamma_\delta$ lying in the $C\delta$-interior of $U$ and each $\theta\in\RR$ one has $\Var(\Re(e^{i\theta}X^v_\tau))\ge c {\delta}^2$, where $X_t^v$ denotes the random walk started at $v$ and $\tau=\tau(B(v,C{\delta}))$ is the first exist time from the ball $B(v,C{\delta})$.
     
     \vskip 2pt
    \item \label{prop:T} The invariant measure $\mu_\delta$ is comparable to the Lebesgue measure starting from the scale $\delta$:
    for each $w\in\CC$ lying in the $C\delta$-interior of $U$ one has $c\delta^2\le \sum_{v\in B(w,{C\delta})}\mu_\delta(v)\le c^{-1}\delta^2$.
  \end{enumerate}
\end{property}

Our next result shows the equivalence of property~\ref{prty:RW} of random walks on Tutte embeddings~$\Gamma_\delta$ to the `geometric' properties~\ref{prty:CONV}/\ref{prty:LIP} of $\Gamma_\delta$ discussed above.

\begin{thma}
\label{thma:LIP<=>RW} (i) Properties \ref{prty:CONV} and/or \ref{prty:LIP} of a Tutte embedding~$\Gamma_\delta$ on a domain $U\subset\CC$ imply property~\ref{prty:RW} on each subdomain $U'\Subset U$ provided that $\delta$ is small enough and with constants~$c,C$ that depend on the constants in \ref{prty:CONV}/\ref{prty:LIP} only.\\[2pt]
(ii) Vice versa, property~\ref{prty:RW} on $U'$ implies properties~\ref{prty:CONV}/\ref{prty:LIP} on each subdomain $U''\Subset U'$ provided that $\delta$ is small enough and with constants $\lambda,\varkappa,C$ that depend on the constants in \ref{prty:RW} only.\\[2pt] Moreover, it is enough to assume that $U'$, resp.~$U''$, lies in an $O(\delta)$-interior of $U$, resp.~$U'$. 
\end{thma}

The fact that~\ref{prty:LIP} implies~\ref{prty:RW} follows from the framework developed in~\cite{CLR1}; in particular,~\cite[Proposition~6.4]{CLR1} gives the uniform ellipticity of the random walk~$X_t$ above scale~$\delta$. We recap this material in Section~\ref{subsec:corollaries_of_Lip} and also sketch a self-contained proof in Remark~\ref{rem:ellipticity}. We prove the converse statement in Section~\ref{sec:LIP<=>RW}; see Theorem~\ref{thmas:RW=>LIP}.

\subsection{Convergence to harmonic functions in a non-trivial metric}%
\label{subsec:when_harmonic}
In general, there exists no coordinate change $w=w(\zeta)$ under which the operator~$\Ll_\vphi$ becomes proportional to the standard Laplacian in the new complex coordinate~$\zeta$. Still, the situations in which such a coordinate change exists are of special interest 
as in this case one recovers the conformal invariance (in the complex structure given by~$\zeta$) of the limits of discrete harmonic functions on~$\Gamma_\delta$. Theorem~\ref{thma:when_harmonic} given below provides a few equivalent conditions that characterize the existence of such a coordinate~$\zeta$.

To formulate this theorem {we need to introduce} a geometric object originating from the main tool that we use in our paper: \emph{t-surfaces}~$\Theta_\delta\subset \RR^{2,2}$ {corresponding} to Tutte embeddings~$\Gamma_\delta$; see Section~\ref{sec:T-embedding of corner graph} for details. As~$\delta\to0$, {these} discrete surfaces~$\Theta_\delta$ converge to a (locally) space-like surface
\begin{equation}
\label{eq:intro-Theta-def}
\Theta = \big\{ (\tfrac{1}{2}(w +\vpsi(w))\,;\,\tfrac12(\overline{\vpsi(w)}-\overline{w}))\in \CC^{1,1}\cong\RR^{2,2}\,\mid\,w\in U \big\},
\end{equation}
{where $\psi=2\vphi_{\bar{w}}=\vphi_x+i\vphi_y$ and} the quadratic form in the Minkowski space $\CC^{1,1}\cong\RR^{2,2}$ is given by
\begin{equation}
\label{eq:intro-R2,2metric-def}
\|(x_1+ix_2\,;x_3+ix_4)\|^2_{1,1}=\|(x_1,x_2\,;x_3,x_4)\|^2_{2,2} = x_1^2+x_2^2-x_3^2-x_4^2\,.
\end{equation}
\begin{thma}
 \label{thma:when_harmonic}
(i) %
The intrinsic metric on the surface~$\Theta$ induced by the metric~\eqref{eq:intro-R2,2metric-def} in $\CC^{1,1}\cong \RR^{2,2}$ 
coincides with the Riemannian metric on the domain $U\subset\CC\cong \RR^2$ defined by the Hessian $D^2\vphi$. \\[2pt]
(ii) The following conditions are equivalent:
  \begin{enumerate}[label=(\alph{enumi})]
    \item The convex potential $\vphi$ solves the Monge--Amp\`ere equation~$\det D^2\vphi=\mathrm{const}$ or, equivalently, its gradient $\vpsi=2\vphi_{\bar{w}}=\vphi_x+i\vphi_y$ preserves the area up to a global multiplicative constant.
    \item $\Theta$ is a maximal space-like surface in $\RR^{2,2}$ (i.e., its mean curvature vanishes at each point).
 
    \item There exists a coordinate change $w=w(\zeta)$ such that %
    we have $\Ll_\vphi h = 0$ if and only if $h$ is harmonic in $\zeta$. If this holds, then $\zeta$ is a conformal coordinate on $\Theta$.    
  \item\label{assertion:L-equiv-Ltimes} %
    The two differential equations $\Ll_\varphi h = 0$ and $\Ll^\times_\varphi h = 0$ are equivalent (that is, have the same space of solutions on each open set $U'\subset U$).
    \item The equation~$\Ll_\varphi \psi=0$ holds (in the weak sense).
    \end{enumerate}
\end{thma}
\begin{rem} Assertion~\ref{assertion:L-equiv-Ltimes} %
  means that %
  {random walks} on the weighted planar graphs in question and on their duals `behave in a similar way' on large scales. (In general, the operators $\Ll_\varphi$ and $\Ll^\times_\varphi$ have the same diffusion parts but different drifts; see Section~\ref{sec:when_harmonic} for details.) One can try to verify this condition in various setups of interest 
as a precursor of the conformal invariance of 2d lattice models (e.g., such as LERW/UST) defined on irregular graphs
\emph{before} embedding these graphs into the complex plane. 
\end{rem}

\subsection{Organization of the paper} 
In Section~\ref{sec:T-embedding of corner graph} we recall the construction of t-embeddings and t-surfaces associated with Tutte embeddings and the modern discrete complex analysis framework developed in~\cite{CLR1}. Section~\ref{sec:Elliptic equations and closed 1-forms in 2D} is devoted to basic properties of elliptic equations in~2d. We prove our main convergence theorems -- Theorem~\ref{thma:harmonic_functions_convergence} and Theorem~\ref{thma:Green_fct_convergence} -- in Section~\ref{sec:Convergence of harmonic functions in C0 topology} starting with classical Dirichlet energy and Caccioppoli estimates. Section~\ref{sec:LIP<=>RW} contains the proof of Theorem~\ref{thma:LIP<=>RW}, that is, the equivalence of properties~\ref{prty:CONV}/\ref{prty:LIP} and property~\ref{prty:RW}. Theorem~\ref{thma:when_harmonic} is proved in Section~\ref{sec:when_harmonic}.

\addtocontents{toc}{\protect\setcounter{tocdepth}{1}}
\subsection*{Acknowledgments} The work of Mikhail Basok was supported by Research Council of Finland grants~333932 and~363549.
The work of Beno\^it Laslier was partially supported by the DIMERS project ANR-18-CE40-0033. Marianna Russkikh was partially supported by a Simons Foundation Travel Support for Mathematicians MPS-TSM00007877. Part of this research was performed while the authors were visiting the Institute for Pure and Applied Mathematics (IPAM), which is supported by the National Science Foundation (Grant No. DMS-1925919).
\addtocontents{toc}{\protect\setcounter{tocdepth}{2}}

\section{T-surface of a harmonic embedding}
\label{sec:T-embedding of corner graph}

Given a weighted planar graph $(\Gamma_\delta,c)$ embedded into the complex plane harmonically, there are two discrete surfaces that can be associated with it. The first is the graph of the Maxwell--Cremona potential $\vPhi_\delta$ discussed in the introduction; this surface belongs to the Euclidean space~$\RR^3$. The second is called a \emph{t-surface}: this is a polygonal surface $\Theta_\delta$ embedded into $\RR^{2,2}\cong\CC^{1,1}$, which is obtained as the lift to this Minkowski space of the \emph{t-embedding} associated with the harmonic embedding as described in~\cite[Section~8.1]{CLR1}. The structure of~$\Theta_\delta$ plays a crucial role in discrete complex analysis techniques developed in~\cite{CLR1} upon which our approach is based. In this section we focus on constructing this t-surface and reviewing basic facts about it. In what follows we drop the index~$\delta$ for shortness.

\subsection{Definition of the t-surface $\Theta_\delta$}
\label{subsec:Definition and basic properties of t-embedding}

Let $(\Gamma,c)$ be a weighted finite planar graph with a chosen outer face. We denote by $\Gamma^\ast$ the dual graph of $\Gamma$ modified as follows: replace the vertex $v_\out^\ast$ of $\Gamma^\ast$ corresponding to the outer face of $\Gamma$ with $n = \deg v_\out^\ast$ vertices $v_{\out,1}^\ast,\dots, v_{\out,n}^\ast$ of degree 1; see Figure~\ref{fig:Gamma_and_its_dual} for an example. Note that we still have a bijection between the edges of $\Gamma$ and $\Gamma^\ast$.

We define the \emph{superposition graph} $\Gamma\cup\Gamma^\ast$ to be the bipartite graph whose vertices correspond to vertices of $\Gamma$, vertices of $\Gamma^\ast$, and midpoints of edges of $\Gamma$, and whose edges correspond to half-edges of $\Gamma$ and $\Gamma^\ast$ (see Figure~\ref{fig:Gamma_and_its_dual}). We call the two classes of vertices black and white and say that a vertex of $\Gamma\cup\Gamma^\ast$ is black if it is a vertex of $\Gamma$ or $\Gamma^\ast$ and white otherwise.

In what follows, we assume that $\Gamma$ is embedded harmonically and denote by $\Hh:\Gamma\to \CC$ the map that sends a vertex of~$\Gamma$ to its position in the complex plane. Let $\Hh^\ast:\Gamma^\ast\to \CC$ be the dual embedding, that is, for each edge $v_1v_2$ and the corresponding dual edge $v_1^\ast v_2^\ast$ we have
\[
  \Hh^\ast(v^\ast_2) - \Hh^\ast(v^\ast_1) = i c_{v_1v_2} (\Hh(v_2) - \Hh(v_1))
\]
if $v_1^\ast$ is on the right of $v_1v_2$.

\begin{figure}
  \centering
  \includegraphics[clip, trim=0cm 0cm 0cm 0.7cm, width = \textwidth]{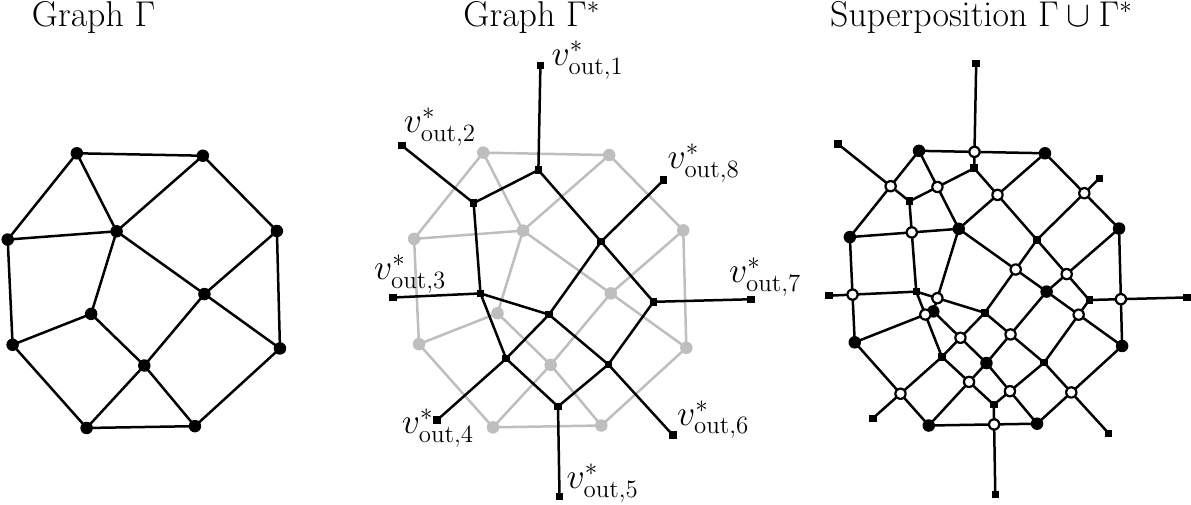}

  \caption{From left to right: an example of a harmonic embedding~$\Hh$ of a weighted graph $\Gamma$, dual harmonic embedding~$\Hh^\ast$ of~$\Gamma^\ast$ (with modified outer vertex), and the superposition graph $\Gamma\cup \Gamma^\ast$. Note that we do \emph{not} fix the embedding of $\Gamma\cup\Gamma^\ast$ into~$\CC$.
  }
  \label{fig:Gamma_and_its_dual}
\end{figure}

Let us emphasize that, unlike in Figure~\ref{fig:Gamma_and_its_dual}, the map $\Hh^\ast$ is not necessary a proper embedding unless one makes additional assumptions such as the convexity of the boundary polygon of~$\Hh$. Even more importantly, $\Hh$ and $\Hh^\ast$ are not necessarily `aligned', i.e., it may be not possible to shift $\Hh^\ast$ so that each vertex $\Hh^\ast(v^\ast)$ belongs to the corresponding face of $\Hh$. In other words, embeddings $\Hh$ and $\Hh^\ast$ do \emph{not} induce an embedding of the superposition graph $\Gamma\cup \Gamma^\ast$ in general. Nevertheless, the pair $(\Hh,\Hh^\ast)$ defines an embedding of the \emph{corner graph} $\Vv$ of $\Gamma$, which coincides with the dual graph of the superposition graph $\Gamma\cup \Gamma^\ast$ away from the boundary: see Definition~\ref{defin:t-embedding_of_Vv} below. We start with a formal definition of~$\Vv=\Vv(\Gamma)$.

\begin{defin}
  \label{defin:corner_graph}
  The corner graph $\Vv$ of $\Gamma$ is the graph whose vertices correspond to pairs $(v,v^\ast)$ of incident vertices of $\Gamma$ and $\Gamma^\ast$. Two vertices $(v_1,v_1^\ast)$ and $(v_2,v_2^\ast)$ are linked by an edge of $\Vv$ if either $v_1 = v_2$ and $v_1^\ast\sim v_2^\ast$ or $v_1\sim v_2$ and $v_1^\ast = v_2^\ast$; see Figure~\ref{fig:corner_graph} for an example.
\end{defin}

Note that there is a bijection between the edges of~$\Vv$ and the edges of~%
$\Gamma\cup \Gamma^\ast$. Also, there is a bijection between faces of $\Vv$ and all vertices of $\Gamma\cup \Gamma^\ast$ except the black boundary ones. The latter correspondence allows us to color the faces of $\Vv$ in black and white according to the coloring of vertices of $\Gamma\cup \Gamma^\ast$; see Figure~\ref{fig:corner_graph} for an example. We denote by $\mB(\Vv)$ and $\mW(\Vv)$ the sets of black and white faces of $\Vv$, respectively.
The next definition follows~\cite[Section~8.1]{CLR1}:

\begin{defin}
  \label{defin:t-embedding_of_Vv}
  T-embedding $\Tt:\Vv\to\CC$ of the corner graph is the map   %
  $\Tt(v,v^\ast) = \frac{1}{2}(\Hh(v) + \Hh^\ast(v^\ast))$.   %
\end{defin}

An example of a t-embedding $\Tt$ is given in Figure~\ref{fig:corner_graph}. It is easy to see that for each black face $b\in \mB(\Vv)$ the polygon $2\Tt(b)$ is a translation of the face of $\Gamma$ or $\Gamma^\ast$ corresponding to $b$. If $w\in \mW(\Vv)$ is a white face of $\Vv$, then $2\Tt(w)$ is a rectangle whose sides are translations of the edge of $\Gamma$ and the edge of $\Gamma^\ast$ corresponding to $w$, and whose orientation agrees with the orientation of the neighboring black faces. From these definitions, it is easy to see that $\Tt$ is always locally proper; moreover, it is also globally proper if we additionally assume that the \emph{boundary polygon} $\partial\Hh$ of $\Hh$ is convex. %

Following the terminology of~\cite{CLR1}, each t-embedding $\Tt$ has an \emph{origami map} $\Oo$ associated with it. Informally speaking, the mapping $\Oo$, which is defined up to rotations and translations, is obtained from the embedding $\Tt$ by folding the plane along all the edges of $\Tt$ consequently (one can prove that this procedure is consistent). Note that the images of a given face of~$\Vv$ in~$\Tt$ and in~$\Oo$ are isometric to each other, and that these images have the same orientation for white faces and the opposite ones for black faces. In the setup of this paper we use the following explicit definition:
\begin{equation}
  \label{eq:def_of_Oo}
  \Oo: \Vv\to \CC,\qquad \Oo((v,v^\ast)) = \tfrac{1}{2}(\overline{\Hh^\ast(v^\ast)} - \overline{\Hh(v)})\,;
\end{equation}
see also equation~\eqref{eq:dO=etadT} below and~\cite[Section~8.1]{CLR1} for more details. 

\begin{figure}
  \centering
  \includegraphics[clip, trim=0cm 0cm 14cm 0.7cm, width = 0.32\textwidth]{Gamma_dual_superposition.pdf}  
  \includegraphics[clip, trim=0cm 0cm 8cm 0.7cm, width = 0.33\textwidth]{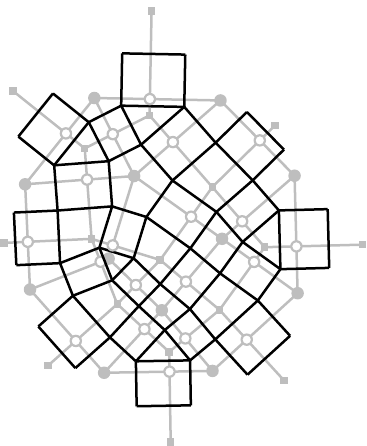}
  \includegraphics[clip, trim=8cm 0cm 0cm 0.7cm, width = 0.33\textwidth]{corner_graph.pdf}

  \caption{From left to right: an example of a harmonic embedding~$\Hh$ of a weighted graph~$\Gamma$, the corner graph $\Vv$ of $\Gamma$,
   and the t-embedding of $\Vv$ constructed out of~$\Hh$.}
  \label{fig:corner_graph}
\end{figure}

The pair of maps $\Tt$ and $\Oo$ gives rise to an embedding of the graph $\Vv$ into the Minkowski space $\CC^{1,1}\cong\RR^{2,2}$. We call this embedding a \emph{t-surface} and denote it by $\Theta$, that is,
\begin{equation}
  \label{eq:def_of_Theta}
  \Theta = (\Tt,\Oo): \Vv\to \CC^{1,1}.
\end{equation}

It is convenient to view $\Tt$ and $\Oo$ as piecewise linear functions on $\Theta$, which in its turn can be viewed as a polygonal surface in $\CC^{1,1}$. It is easy to see that on each face of~$\Theta$ we have
\begin{equation}
  \label{eq:|dT|=|dO|}
  |d\Tt| = |d\Oo|;
\end{equation}
see also equation~\eqref{eq:dO=etadT} below. Recall that the function $\vPhi$, defined by~\eqref{eq:intro-Phi-def}, is the Maxwell--Cremona potential associated with $\Gamma$ and that $\vPsi = 2\partial_{\bar w} \vPhi=\partial_x\vPhi+i\partial_y\vPhi$. 

\begin{lemma}
  \label{lemma:T-barO_projection}
  (i) If $b\in \mB(\Vv)$ corresponds to a vertex $v$ of $\Gamma$ and $p\in \Theta(b)$, then $(\Tt - \overline \Oo)(p)=\Hh(v)$.\\[2pt]
  (ii) If $u\in \mW(\Vv)$ corresponds to an edge of~$\Gamma$ and $p\in\Theta(u)$, then $(\Tt-\overline \Oo)(p)$ lies on the image of this edge of~$\Gamma$ in the harmonic embedding~$\Hh$.\\[2pt]
  (iii) Let $b\in \mB(\Vv)$ corresponds to a face of $\Gamma$. Assume that $p\in \Theta(b)$ and let $w = (\Tt - \overline \Oo)(p)$. Then, $w$ belongs to the image of this face of~$\Gamma$ in~$\Hh$ and one has
  \begin{equation}
  \label{eq:T-barO_projection}
    \Tt(p) = \tfrac{1}{2}(w + \vPsi(w))\,,\qquad \Oo(p) = \tfrac{1}{2}(\overline{\vPsi(w) -w})\,.
  \end{equation}
\end{lemma}

\begin{proof}
This follows directly from the definition of $\Tt$ and $\Oo$; see also~\cite[Section~8.1]{CLR1}.
\end{proof}

Let $\Ww=\Tt-\overline\Oo$. According to Lemma~\ref{lemma:T-barO_projection}, the mapping $\Ww$ projects $\Theta$ onto the union of faces of the harmonic embedding $\Hh$ of~$\Gamma$. As already mentioned above, the t-embedding $\Tt$ is not always globally proper (it may have overlaps if the boundary polygon~$\partial \Hh$ of the harmonic embedding~$\Hh$ is not convex), which means that $\Tt$ is not necessarily an injective map from the t-surface~$\Theta$ to~$\CC$. However, 
$\Tt$ is always locally one-to-one. Using~\eqref{eq:|dT|=|dO|}, it is easy to make this statement quantitative:

\begin{lemma}
  \label{lemma:T_is_locally_one-to-one}
  Let $p\in \Theta$ be an arbitrary point on the t-surface and $2r \le \dist(\Ww(p),\partial\Hh)$. There is an open (in the topology induced on~$\Theta$ from~$\CC^2$) connected subset $p\in B_\Theta(p,r)\subset\Theta$ such that $\Tt:B_\Theta(p,r)\to B(\Tt(p), r)$ %
  is a bijection. Moreover, $\Ww(B_\Theta(p,r))\subset B(\Ww(p), 2r)$.
\end{lemma}

\begin{proof} Let $B_\Theta(p,r)$ be the connected component of $\Tt^{-1}(B(\Tt(p), r))$ that contains $p$.
  For a point $w\in B(\Tt(p), r)$, let $\ell_w$ be the straight segment connecting $\Tt(p)$ and $w$ and let $\ell'_w\subset \ell$ be the longest subsegment that contains $\Tt(p)$ and does not intersect $\Tt(\partial\Theta)$. Denote by~$\gamma_w\subset\Tt$ the lift of $\ell'_w$ onto the t-surface. From the fact that $\Tt$ is locally proper it is easy to see that $\Tt$ is one-to-one on~$\gamma_w$. 
  It follows from~\ref{eq:|dT|=|dO|} that $\diam(\Ww(\gamma_w))\leq 2r$. If $\ell'_w\subsetneq\ell_w$, we would also have $\Ww(\gamma_w)\cap \partial \Hh\neq\varnothing$, a contradiction. Let~$B_\Theta(p,r)$ denote the union of paths $\gamma_w$ as above. By construction, $\Tt$ is a bijection between~$B_\Theta(p,r)$ and~$B(w,r)$. As~$\diam(\Ww(\gamma_w))\leq 2r$, we have $\Ww(B_\Theta(p,r))\subset B(\Ww(p),2r)$.
\end{proof}

\subsection{%
Property~$\Lip(\kappa,\delta)$ of the t-surface}
\label{subsec:equivalence_of_Lip}
According to Lemma~\ref{lemma:T-barO_projection}, we can view $\Tt$ and $\Oo$ as the pullbacks of the functions $w\mapsto \tfrac12(w + \vPsi)$ and $w\mapsto \tfrac12(\overline{\vPsi(w) - w})$ from the plane into which $\Gamma$ is embedded by~$\Hh$ to the corresponding t-surface $\Theta$. (Contrary to~$\vPsi$, the maps $\Tt$ and~$\Oo$ are well-defined and continuous on $\Theta$.) If $\vPhi$ was a \emph{smooth} and uniformly convex function, then  $\Oo=\overline{\partial}_w\vPhi-\frac12\overline{w}$ would be, locally, a $\kappa$-Lipschitz function of $\Tt=\frac12w+\partial_{\bar{w}}\vPhi$ for some $\kappa<1$ depending only on the constant~$\lambda$ in~\eqref{eq:Conv} or, equivalently, on the constant~$\varkappa$ in~\eqref{eq:Lip}; see also Lemma~\ref{lemma:Lip_equiv_conv}.
Moreover, one can easily see that $\vPhi$ is uniformly convex if and only if this condition holds with some $\kappa<1$. In our discrete setup, $\Oo$ is locally a linear function of $\Tt$, which means that this $\kappa$-Lipschitzness property cannot hold at very small distances. However, it is reasonable to expect that property~\ref{prty:CONV}/\ref{prty:LIP} implies that $\Oo$ is a $\kappa$-Lipschitz function of~$\Tt$ starting from the scale $\delta$. The following definition appeared in~\cite{CLR1}:

\begin{defin}
  \label{defin:Lip}
  Let $\delta>0$ and $0<\kappa < 1$ be given. We say that a t-surface $\Theta = (\Tt,\Oo)$ satisfies property $\Lip(\kappa,\delta)$ if for each set $U_\Theta\subset \Theta$ such that $\Tt$ restricted to $U_\Theta$ is one-to-one and $\Tt(U_\Theta)$ is convex and for each $p_1,p_2\in U_\Theta$ one has 
  \begin{equation}
  \label{eq:LipKdelta}
    |\Oo(p_2) - \Oo(p_1)|\leq\kappa|\Tt(p_2) - \Tt(p_1)|\ \ \text{if}\ \  |\Tt(p_2) - \Tt(p_1)|\geq \delta.
  \end{equation}
\end{defin}

The next lemma shows that property $\Lip(\kappa,\delta)$ is equivalent to property~\ref{prty:LIP} in the same sense as~\ref{prty:LIP} is equivalent (see Lemma~\ref{lemma:Lip_equiv_conv}) to property~\ref{prty:CONV}, i.e., $O(\delta)$-away from the boundary.

\begin{lemma}
  \label{lemma:Lip_follows_from_our_assumptions}
  Let~$\Theta$ be the t-surface corresponding to~$\Hh$ and~$\Psi=2\partial_{\bar{w}}\Phi$ as above.\\[2pt]
  (i) Assume that~$\Theta$ satisfy property $\Lip(\kappa,\delta)$. Then, $\Psi$ has property~\ref{prty:LIP} with $\varkappa=\frac{1-\kappa}{1+\kappa}$ and~$C=2$ provided that the points~$w_1,w_2$ in~\eqref{eq:Lip} are $8(1-\kappa)^{-1}\delta$-away from the boundary polygon $\partial\Hh$ of~$\Hh$.\\[2pt]
  (ii) Vice versa, if~$\Psi$ has property~\ref{prty:LIP}, then $\Theta$ satisfies~$\Lip(\kappa,C'\delta)$ for each $\kappa>\frac{1-\varkappa}{1+\varkappa}$ provided that the points~$p_1,p_2$ in~\eqref{eq:LipKdelta} are $C'\delta$-away from~$\partial\Theta$, where $C'$ depends on $\varkappa, \kappa$, and $C$ in~\ref{prty:LIP} only.
\end{lemma}

\begin{proof} (i) Let $r=4(1-\kappa)^{-1}\delta$. Recall the notation~$\Ww=\Tt-\overline{\Oo}$ and note that Lemma~\ref{lemma:T-barO_projection} yields
\begin{equation}
   \label{eq:lfa0}
\Psi\circ\Ww=\Tt+\overline{\Oo}
\end{equation}
if we restrict these maps to the black faces of $\Theta$ corresponding to faces of~$\Gamma$. Let~$p_1\in\Theta$, $z_1=\Tt(p_1)$, and $w_1=\Ww(p_1)$ be such that $\dist(w_1,\partial\Hh)\ge 2r$. Consider the neighborhood~$B_\Theta=B_\Theta(p_1,r)$ from Lemma~\ref{lemma:T_is_locally_one-to-one}; recall that~$\Tt:B_\Theta\to B(z_1,r)$ is a bijection and that $\Ww(B_\Theta)\subset B(w_1,2r)$. It is also easy to see that the Lipschitz property~\eqref{eq:LipKdelta} implies that~$\Ww(B_\Theta)\supset B(w_1,4\delta)$ since the image of the circle $|z-z_1|=r=4(1-\kappa)^{-1}\delta$ under the map~$\Ww\circ\Tt^{-1}$ stays at least $4\delta$-away from~$w_1$ and encircles~$w_1$ due to topological reasons.

Let~$w_2\in B(w_1,4\delta)\smallsetminus B(w_1,2\delta)$; without true loss of generality we can also assume that~$w_1$ and~$w_2$ lie in the interiors of faces of~$\Hh$. Denote~$p_2=\Ww^{-1}(w_2)\in B_\Theta$ and $z_2=\Tt(p_2)\in B(z_1,r)$. Note that $|z_2-z_1|\ge\delta$ as the mapping~$\Ww\circ\Tt^{-1}$ is $2$-Lipschitz. It is now easy to see from~\eqref{eq:lfa0} and~\eqref{eq:LipKdelta} that inequalities~\eqref{eq:Lip} hold for these~$w_1$ and~$w_2$ and~$\varkappa=(1-\kappa)/(1+\kappa)$. The general case follows by linearity: split the segment $[w_1;w_2]$ into subsegments of length between~$2\delta$ and~$4\delta$.

\smallskip

\noindent (ii) Let~$r=C'\delta$, where $C'>C(1-\kappa)^{-1}$ is a large enough constant. First, assume that~$p_1\in\Theta$ belongs to a black face corresponding to a face of~$\Gamma$ and $\dist(\Tt(p_1),\partial\Hh)\ge 2r$. Let~$p_2\in B_\Theta=B_\Theta(p_1,r)$ be another such point 
with~$|\Tt(p_2)-\Tt(p_1)|\ge \frac12r$. Denote~$w_1=\Ww(p_1)$, $w_2=\Ww(p_2)$ as above, and recall that~$w_2\in\Ww(B_\Theta)\subset  B(w_1,2r)$. It follows from~\eqref{eq:lfa0} and the upper bound in~\eqref{eq:Lip} that the image of the circle $|w-w_1|=C\delta$ under the map~$\Tt\circ\Ww^{-1}$ lies inside the disc $B(p_1,C(1-\varkappa)^{-1}\delta)$. Therefore, $|w_2-w_1|\ge C\delta$. It is now easy to deduce from~\eqref{eq:T-barO_projection} and the lower bound in~\eqref{eq:Lip} that
\[
\biggl|\frac{\Oo(p_2)-\Oo(p_1)}{\Tt(p_2)-\Tt(p_1)}\biggr|\,=\,\biggl|\frac{(\vPsi(w_2)-\vPsi(w_1))-(w_2-w_1))}{(w_2-w_1)+(\vPsi(w_2)-\vPsi(w_1))}\biggr| \,\le\, \frac{1-\varkappa}{1+\varkappa}\,.
\]

It remains to consider points $p_1,p_2\in\Theta$ that does not necessarily belong to black faces of~$\Theta$ that correspond to faces of~$\Gamma$ but still satisfy $p_2\in B_\Theta(p_1,r)$ and $|\Tt(p_2)-\Tt(p_1)|\ge \frac12r$. Replacing these points by nearby points on such faces and using~\eqref{eq:Lip} again, it is easy to see that
\[
|\Oo(p_2)-\Oo(p_1)|\ \le\ \frac{1-\varkappa}{1+\varkappa}\,|\Tt(p_2)-\Tt(p_1)|+2C(1+\varkappa^{-1})\delta\,.
\]
Therefore, for each~$\kappa>\tfrac{1-\varkappa}{1+\varkappa}$ one can find a large enough constant~$C'$ such that the estimate~\eqref{eq:LipKdelta} holds without additional assumptions on~$p_1,p_2$ provided that $|\Tt(p_2)-\Tt(p_1)|\ge C'\delta$.
\end{proof}

\subsection{Corollaries of property $\Lip(\kappa,\delta)$ for random walks and harmonic functions on~$\Gamma$.} 
\label{subsec:corollaries_of_Lip}
In order to simplify the notation, from now onwards we usually identify vertices~$v$ of an abstract weighted planar graph~$(\Gamma,c)$ with their positions~$\Hh(v)$ in the complex plane if no confusion arises.

\begin{prop}
  \label{prop:RW_ellipticity}
  Let $\Gamma$ be a harmonic embedding and the t-surface $\Theta$ be defined as above. Let~$X_t$ be the continuous time simple random walk on~$\Gamma$ parameterized so that $|X_t|^2-t$ is a local martingale. Denote by $X_t^v$ the random walk started at $X_0^v = v$ and stopped at~$\partial \Gamma$. Assume that $\Theta$ has property $\Lip(\kappa,\delta)$ for some $0<\kappa <1$ and $\delta>0$ (or, equivalently, that $\Gamma$ has properties~\ref{prty:CONV}/\ref{prty:LIP}). Then, there is a constant $C=C(\kappa)>0$ such that for each $t\geq C\delta^2$, each vertex~$v$ of~$\Gamma$ with $\dist (v,\partial\Gamma)\ge \sqrt{t}$, and each $\theta\in \RR$ we have
  \begin{equation}
    \label{eq:RW_ellipticity}
    \Var(\Re(e^{i\theta}X_t^v)) \geq C^{-1}t.
  \end{equation}
  The same assertion holds for the dual graph $\Gamma^\ast$.
\end{prop}
\begin{proof}
  As explained in~\cite[Section~8.1]{CLR1}, $\Gamma$ and $\Gamma^\ast$ are \emph{T-graphs} obtained from $\Theta = (\Tt,\Oo)$ as projections $\Tt - \overline\Oo$ and $\Tt + \overline\Oo$ (cf. Lemma~\ref{lemma:T-barO_projection}). The ellipticity estimate~\eqref{eq:RW_ellipticity} for the random walk on such T-graphs follows from~\cite[Proposition~6.4]{CLR1}. More formally, since~\eqref{eq:RW_ellipticity} is local and the surface $\Theta$ satisfies condition $\Lip(\kappa,\delta)$, we can assume that the projection $\Theta\mapsto \Tt$ is one-to-one (cf. the proof of Lemma~\ref{lemma:Lip_follows_from_our_assumptions}), that is, $\Tt$ is a proper t-embedding as defined in~\cite{CLR1} and $\Oo$ is a concrete instance of its origami map. If we now swap the colors of faces of $\Tt$, then $\pm\overline \Oo$ will represent the origami map of thus obtained t-embedding (see~\cite[Section~2.2]{CLR1}), to which we can apply~\cite[Proposition~6.4]{CLR1} in order to get the desired estimates for $\Tt - \overline\Oo$ and $\Tt + \overline\Oo$.
\end{proof}

\begin{rem}
  \label{rem:RW_ellipticity_self-contained0}
  Instead of referring to~\cite[Proposition~6.4]{CLR1} one can give a self-contained proof of Proposition~\ref{prop:RW_ellipticity} using the notion of discrete extremal length; see Remark~\ref{rem:ellipticity} for further comments.
\end{rem}

The following lemmas are standard corollaries of Proposition~\ref{prop:RW_ellipticity}; see~\cite[Sections~6.3,~6.4]{CLR1}):

\begin{lemma}[crossing estimates]
  \label{lemma:crossing_estimates}
  Under the assumptions of Proposition~\ref{prop:RW_ellipticity}, there exists a constant $C=C(\kappa)>0$ such that for each $r\geq C\delta$, each vertex $v$ of $\Gamma$ with $\dist(v,\partial\Gamma)\ge Cr$, and each $\theta_0\in \RR$ we have
  \[
    \PP\biggl[\begin{array}{l}\text{random walk $X_t^v$ exists the disc $B(v,r)$}\\ \text{through the arc }\{v\!+\!e^{i\theta}r,\ |\theta-\theta_0|\leq \pi/4\}\end{array}\biggr]\,\geq\, C^{-1}.
  \]
  The same assertion holds for the dual graph $\Gamma^\ast$.
\end{lemma}

\begin{lemma}[Harnack inequality]
  \label{lemma:Harnack}
  Under the assumptions of Proposition~\ref{prop:RW_ellipticity}, there exists a constant $C=C(\kappa)>0$ such that for each $r\geq C\delta$, each vertex $v$ of $\Gamma$ with $\dist(v,\partial\Gamma)\ge 2r$, and each non-negative discrete harmonic function $H$ on $B(v,2r)$ we have
  \[
    \max_{v\in B(v,r)} H(v)\,\leq\, C\min_{v\in B(v,r)} H(v).
  \]
  The same assertion holds for the dual graph $\Gamma^\ast$.
\end{lemma}

\begin{lemma}[H\"older-type decay of oscillations]
  \label{lemma:Holder}
  Under the assumptions of Proposition~\ref{prop:RW_ellipticity} there exist constants $\beta=\beta(\kappa)>0$ and $C=C(\kappa)>0$ such that for each $R\geq r\geq C\delta$, each vertex $v$ of $\Gamma$ with $\dist(v,\partial\Gamma)\ge CR$, and each discrete harmonic function $H$ on $B(v,R)$ we have
  \[
    \osc_{B(v,r)}H \,\leq\, C\left( \frac{r}{R} \right)^\beta \osc_{B(v,R)} H.
  \]
  The same assertion holds for the dual graph $\Gamma^\ast$.
\end{lemma}

It follows from Lemma~\ref{lemma:crossing_estimates} that property $\Lip(\kappa,\delta)$ of the t-surface~$\Theta$ implies item~(a) in property \ref{prty:RW} of the random walks on~$\Gamma$. Let us show that $\Lip(\kappa,\delta)$ also implies item~(b) in~\ref{prty:RW}.

\begin{lemma}
  \label{lemma:white_area_is_bdd}
  Let $\Gamma$ be a harmonic embedding and the t-surface~$\Theta$ be defined as above. Assume that $\Theta$ has property $\Lip(\kappa,\delta)$ for some $0<\kappa < 1$ and $\delta>0$.
   There exists a constant $C=C(\kappa)>0$ such that for each $r\geq C\delta$ and each vertex~$v$ of~$\Gamma$ with $\dist(v,\partial\Gamma)\ge Cr$ we have
  \[
    C^{-1}r^2\, \leq\, \sum\nolimits_{v_1,v_2\in \Gamma\cap B(v,r):\ v_1\sim v_2} c_{v_1v_2} |v_1 - v_2|^2\,\leq\, Cr^2.
  \]
\end{lemma}
\begin{proof}
  Since the required estimate is local, it is enough to assume that the projection $\Theta\mapsto \Tt$ is one-to-one; see Lemma~\ref{lemma:T_is_locally_one-to-one}). Given an edge $v_1v_2$ of $\Gamma$ denote by $u(v_1v_2)\in\mW(\Vv)$ the corresponding white face of $\Vv$. It follows from Definition~\ref{defin:t-embedding_of_Vv} that
  \[
    c_{v_1v_2}|v_1-v_2|^2 = 4\Area(\Tt(u(v_1v_2)))\,.
  \]
  The claim now follows from~\cite[Eq.~(6.1) and Lemma~6.3]{CLR1}.
\end{proof}

\subsection{T-holomorphic functions}
\label{subsec:T-holomorphic functions}

Let $\Gamma$ be a weighted planar graph and $\Gamma^\ast$ be its dual. There is a standard way to define a discrete holomorphic function on the superposition of $\Gamma$ and $\Gamma^\ast$ by declaring it to be a pair of two functions $H: \Gamma\to \RR$ and $H^\ast: \Gamma^\ast\to \RR$ that are harmonic conjugate to each other. Assume now that $\Gamma$ is embedded harmonically and $\Tt$ is the corresponding t-embedding; see Section~\ref{subsec:Definition and basic properties of t-embedding}. In this case the aforementioned discrete holomorphic functions correspond to \emph{t-white-holomorphic functions} on $\Tt$ as defined in~\cite{CLR1}. It is convenient to introduce an alternative definition of the latter using the terminology developed in Section~\ref{subsec:Definition and basic properties of t-embedding}.

We begin by fixing an instance of the \emph{origami square root function} $\eta$ associated with the t-embedding $\Tt$ and the origami map $\Oo$; see~\cite[Section~2.2]{CLR1}. The function $\eta$ is defined on the set $\mB(\Vv)\cup \mW(\Vv)$ of faces of~$\Tt$. Recall that~$\Vv$ stands for the corner graph of~$\Gamma$ and that there is a correspondence between $\mB(\Vv)$ on the one side and vertices of $\Gamma$ and those of $\Gamma^\ast$ on the other, as well as between $\mW(\Vv)$ and edges of $\Gamma$, which are also in a one-to-one correspondence with edges of~$\Gamma^\ast$. (See Definition~\ref{defin:corner_graph} and the discussion below it.) 
Define the function~$\eta$ as follows:
\begin{equation}
  \label{eq:def_of_eta}
  \begin{array}{l} 
  \eta_b = 1\ \text{if}\ b\in \Gamma\subset \mB(\Vv),\\
  \eta_b = i\ \text{if}\ b\in \Gamma^\ast\subset\mB(\Vv),
  \end{array}\quad 
  \eta_u %
    = \pm i \frac{\overline{v_1 - v_2}}{|v_1 - v_2|}\ \ \text{if}\ u\in\mW(\Vv)\text{ corresponds to an edge $v_1v_2$ of $\Gamma$,}
\end{equation}
where the $\pm$ signs are chosen arbitrary. %
This function helps to relate the gradients of $\Tt$ and $\Oo$ viewed as piecewise linear mappings on the t-surface $\Theta=(\Tt,\Oo)$ corresponding to~$\Tt$: 
\begin{equation}
  \label{eq:dO=etadT}
  d\Oo(p) = \begin{cases}
    \eta_u^2\,d\Tt(p)&\text{if $p$ belongs to a white face }\Theta(u),\ u\in \mW(\Vv),\\
    \overline\eta_b^2\,d\overline{\Tt(p)} & \text{if $p$ belongs to a black face }\Theta(b),\ b\in \mB(\Vv).
  \end{cases}
\end{equation}

Given a function $F:\mB(\Vv)\to \CC$, let $F\,d\Tt$ be the 1-form defined on edges of $\Theta$, which equals $F(b)\,d\Tt$ on each edge incident to a face $\Theta(b)$, $b\in \mB(\Vv)$, of the t-surface~$\Theta$.

\begin{defin}
  \label{defin:t-white-holomorphic}
  A function $F: \mB(\Vv)\to \CC$ is called t-white-holomorphic if for each $b\in \mB(\Vv)$ we have $F(b)\in \eta_b\RR$ and the 1-form $F\,d\Tt$ defined on edges of $\Theta$ is closed.
\end{defin}

The assertion that the 1-form $F\,d\Tt$ is closed implies a non-trivial relation on $F$ for each $u\in \mW(\Vv)$. Namely, let $u=u(v_1v_2)$ corresponds to an edge $v_1v_2$ of~$\Gamma$. As usual, we denote by $v_1^\ast v_2^\ast$ the dual edge of~$\Gamma^\ast$ oriented such that $v_1^\ast$ is on the right of $v_1v_2$. By the definition of the t-embedding~$\Tt$ (see Fig.~\ref{fig:corner_graph}), the fact that the form $F\,d\Tt$ is closed around~$u$ can be written as the equation
\[
(-F(v_1)+F(v_2))(\Hh^\ast(v_2^\ast) - \Hh^\ast(v_1^\ast)) +(F(v_1^\ast)-F(v_2^\ast))(\Hh(v_2) - \Hh(v_1)) = 0\,,
\]
which is equivalent to saying that 
\[
  F(v_2^\ast) - F(v_1^\ast) = ic_{v_1v_2} (F(v_2) - F(v_1)).
\]
Taking into account that $F(b)$ is purely real when $b\in \Gamma$ and purely imaginary when $b\in \Gamma^\ast$ we conclude that the functions $H = F\vert_\Gamma$ and $H^\ast = -i F\vert_{\Gamma^\ast}$ are harmonic conjugate to each other. Thus, we obtained the definition of discrete holomorphicity that we have started with. We address the reader to~\cite[Section~8.1]{CLR1} for more comments.

By the definition given above, a t-white-holomorphic function $F$ is a pair $(H, iH^\ast)$ of harmonic functions, whereas, by the analogy with the continuous holomorphicity, the actual holomorphic function should correspond to the sum $H + iH^\ast$. This expression does not make sense `as is' because the two functions~$H$ and~$H^\ast$ are defined on different sets. However, we can give it a meaning by extending~$F$ to white, rectangular, faces of~$\Vv$ as follows. 

\def\Wws{\mW_\spl}
\def\Bbs{\mB_\spl}

\begin{defin}
  \label{defin:true_complex_values}
  For each white face of $\Vv$ choose its diagonal splitting into two triangles and denote by $\Wws(\Vv)$ the set of thus obtained triangles. Given a t-white-holomorphic function $F$, define a function $F^\circ: \Wws(\Vv)\to \CC$ so that for each $u\in \Wws(\Vv)$ incident to a black face $b\in \mB(\Vv)$ we have
  \[
    F(b) = \Pr(F^\circ(u),\eta_b\RR).
  \]
\end{defin}

Note that each triangle $u\in \Wws(\Vv)$ has exactly two black faces $v\in \Gamma$ and $v^\ast\in \Gamma^\ast$ incident to~it. 
By definition, $F(v)\in \eta_v \RR = \RR$ and $F(v^\ast)\in \eta_{v^\ast}\RR = i\RR$. Thus, we simply have
\[
  F^\circ(u) = F(v) + F(v^\ast)\ \ \text{if}\ \ v\sim u\sim v^\ast.
\]
Following~\cite{CLR1}, we call these values \emph{`true complex values'} of a t-white-holomorphic function $F$. 

\begin{lemma}
  \label{lemma:closed_form}
  Let $F:\mB(\Vv)\to\CC$ be a t-white-holomorphic function and $F^\circ$ be its extension to~$\Wws(\Vv)$ defined above. Define a piecewise constant 1-form $F^\circ \,d\Tt + \overline {F^\circ}\,d\overline \Oo$ on the t-surface~$\Theta$ as follows:
  \[%
  (F^\circ \,d\Tt + \overline{F^\circ}\,d\overline \Oo)(p) =
  \begin{cases}
  F^\circ(u) \,d\Tt(p) + \overline{F^\circ(u)}\,d\overline{\Oo(p)} & \text{if $p\in \Theta(u)$, $u\in\Wws(\Vv)$,}\\
  2F(b)\,d\Tt(p) = F^\circ(u) \,d\Tt(p) + \overline{F^\circ(u)}\,d\overline {\Oo(p)} & \text{if $p\in\Theta(b)$, $b\in\Bb(\Vv)$},
  \end{cases}
  \]%
  where in the second line one can take any white triangle~$u$ incident to the face~$b$. This piecewise constant form is well defined, does not depend on the choices of~$u\sim b$, and is closed.
\end{lemma}
\begin{proof}
See~\cite[Proposition~3.7 and Section~5]{CLR1}. The choice of~$u\sim b$ is irrelevant due to~\eqref{eq:dO=etadT}. The fact that this differential form is well defined on the diagonals that were used to split white faces is equivalent to the fact that $H$ and $H^\ast$ are harmonic conjugate to each other.
\end{proof}

In Section~\ref{subsec:C1 under ExpFat} and Section~\ref{subsec:Extension to C1,1} we also use the notion of \emph{t-black-holomorphic functions} on~$\Theta$. Similarly to Definition~\ref{defin:t-white-holomorphic}, a t-black-holomorphic function $G:\mW(\Vv)\to\CC$ is originally defined on white (rectangular) faces of~$\Theta$ so that~$G(u)\in\eta_u\RR$ for each~$u$ and the differential 1-form $Gd\Tt$ defined on edges of~$\Theta$ as $G(u)d\Tt$ is closed. It follows from~\eqref{eq:def_of_eta} that the imaginary part of the primitive of this form is constant on vertices $(vv^\ast)\in\Vv$ that share the  
same vertex $v\in\Gamma$ and similarly for the real part and $v^\ast\in \Gamma^\ast$. This allows one to consider the primitives 
\[
\textstyle H(v):=\int^v \Im[Gd\Tt],\qquad H^\ast(v^\ast):=-\int^{v^\ast}\!\Re[Gd\Tt].
\]
By construction,~$H$ and~$H^\ast$ are harmonically conjugate to each other. In particular, $H$ is harmonic on~$\Gamma$ and $H^\ast$ is harmonic on~$\Gamma^\ast$. In other words, t-black-holomorphic functions on~$\mW(\Vv)$ are, up to the multiple~$i$, the \emph{gradients} of harmonic functions on~$\Gamma$.

In order to define `true complex values' $G^\bullet$ of~$G$, we triangulate each black face of~$\Theta$, extend $H$ and~$H^\ast$ linearly to thus obtained triangles $b\in \Bbs(\Vv)$, and define $G^\bullet(b)$ to be the gradient of~$H$ (if~$b$ is a part of a face of $\Gamma^\ast$) or the gradient of~$H^\ast$ (if~$b$ is a part of a face of $\Gamma$), viewed as a complex number and multiplied by $i$. Similarly to Lemma~\ref{lemma:closed_form}, this allows one to introduce a closed differential 1-form~$G^\bullet d\Tt+\overline{G^\bullet}d\Oo$ on $\Theta$. We address the reader to~\cite[Section~8.1]{CLR1} for more details.

\section{Elliptic equations and closed 1-forms in 2D}
\label{sec:Elliptic equations and closed 1-forms in 2D}

In this section we collect some facts about the linearized Monge--Amp\`ere equation that we need later. This is a classical topic at the crossroad of the theory of elliptic PDE and the optimal transport problem, so we do not attempt to summarize the literature and only mention the foundational paper~\cite{caffarelli-gutierrez-1997-lin-MA} and the monograph~\cite{figalli-2017-MA} where the link with the transport problem is addressed. For our purposes it is enough to restrict the discussion to linearized Monge--Amp\`ere equations with uniformly convex potentials. Such equations are uniformly elliptic which makes the standard theory of elliptic PDE applicable. We address the reader to~\cite{kenig2000new} and references therein for the standard properties of elliptic operators in divergence form with measurable coefficients. The book~\cite{AstalaBook} is as a standard reference for the interplay between Beltrami and elliptic equations in 2D.

We begin with the notion of $A$-harmonic conjugate that can be related with any elliptic operator in 2D in divergence form. Let $\Omega\subset \CC$ be a simply connected domain and $A$ be a $2\times 2$ matrix whose coefficients are measurable real-valued functions in $\Omega$. In what follows we use the notation $w = x+iy$ for the coordinate in $\CC$. %
Consider a differential operator
\begin{equation*} %
  \Ll h = -\div(A\nabla h).
\end{equation*}
We assume that the operator $\Ll$ has uniformly bounded coefficients and that $\Ll$ is strongly elliptic: that is, there exists a constant $\lambda>0$ such that for each $w\in \Omega$ and $\nu \in \RR^2$ we have
\begin{equation}
  \label{eq:strong_ellipticity}
  \lambda|\nu|^2\ \leq\ \nu\cdot A(w)\nu\ \leq\ \lambda^{-1}|\nu|^2.
\end{equation}
Let us emphasize that we do not make any smoothness assumption on the coefficients of $A$. In particular, the equation $\Ll h = 0$ is understood in a weak form, that is, we say that $\Ll h = 0$ if the gradient of $h$ is locally square integrable and for each test function $\xi\in \mC^\infty_0(\Omega)$ we have
\[
  \int_\Omega (\nabla \xi\cdot A\nabla h)\,dxdy= 0.
\]

Denote by $\ast = \left(\begin{smallmatrix} 0 & -1 \\ 1 & 0\end{smallmatrix}\right)$ the Hodge star operator in $\RR^2$.
Classically~\cite[Chapter~16.1.3]{AstalaBook}, with each solution $h$ of the equation $\Ll h = 0$ we can associate its $A$-harmonic conjugate $h^\ast$ defined by
\begin{equation} 
\label{eq:def_of_A-conj}
  \nabla h^\ast = \ast A\nabla h.
\end{equation}
(The equation $\Ll h = 0$ is equivalent to $\ast A\nabla h$ being curl-free, that is to $\ast A\nabla h$ being the gradient of a function.) Also, note that
\[
  \div(\ast A^{-1} \ast \nabla h^\ast) = -\div(\ast\nabla h) = 0
\]
hence $h^\ast$ is annihilated by the `dual' operator
\[
  \Ll^\times h = \div(\ast A^{-1} \ast \nabla h).
\]
For example, if $A = \Id_{2\times 2}$, then $\Ll$ is the (positively defined) Euclidean Laplacian and $h^\ast$ is the usual harmonic conjugate of a harmonic function $h$. In this case $f = h + ih^\ast$ is holomorphic and the relation~\eqref{eq:def_of_A-conj} is nothing but the Cauchy--Riemann equations. 

Let us now focus on elliptic operators $\Ll_\vphi$ that give rise to the linearized Monge--Amp\`ere equation. In this case, the matrix $A$ is symmetric and takes the form
\begin{equation}
  \label{eq:A_no_drift}
  A  \ = \ \begin{pmatrix} \vphi_{yy} & -\vphi_{xy} \\ -\vphi_{xy} & \vphi_{xx} \end{pmatrix} \ =\ \begin{pmatrix} \Re (\vpsi_w - \vpsi_{\bar w}) & \Im (\vpsi_w - \vpsi_{\bar w})\\ -\Im(\vpsi_w + \vpsi_{\bar w}) & \Re(\vpsi_w + \vpsi_{\bar w})\end{pmatrix}\ = \ \begin{pmatrix} v_y & -u_y\\ -v_x & u_x\end{pmatrix},
\end{equation}
where $\vpsi = u + iv = 2\partial_{\bar w}\vphi$ and $\vphi$ is a convex function. The boundedness of entries of $A$ is equivalent to the Lipschitzness of $\vpsi$ while the ellipticity condition~\eqref{eq:strong_ellipticity} is equivalent to the existence of $a>0$ and $\mu < 1$ such that the inequalities
\begin{equation}
  \label{eq:vpsi_ellip}
  \vpsi_w\geq a,\qquad |\vpsi_{\bar w}|\leq \mu \vpsi_w
\end{equation}
hold almost everywhere. (Note that $\vpsi$ is differentiable almost everywhere by Rademacher's theorem.) Further, \eqref{eq:vpsi_ellip} is equivalent to $\vphi$ being uniformly convex, that is, to the existence of $\lambda>0$ such that for each segment $[w_1,w_2]\subset \Omega$ we have
\begin{equation}
  \label{eq:vphi_uniformly_convex}
  \lambda|w_1 - w_2|^2 \leq \vphi(w_2) - 2\vphi(\tfrac12(w_1 + w_2)) + \vphi(w_1) \leq \lambda^{-1}|w_1 - w_2|^2.
\end{equation}

\begin{lemma}
  \label{lemma:Morera_condition}
  Let $\vphi$ be a uniformly convex function in a domain $\Omega\subset\CC$ and $\vpsi = 2\partial_{\bar w}\vphi$. Assume that $h,h^\ast$ are continuous real-valued functions in $\Omega$ such that the differential form $h\,d\vpsi+ih^\ast\,dw$ is closed `in the weak sense', i.e., that its integral over each contractible piecewise smooth loop vanishes. 
  \begin{enumerate}[label=(\roman*)]
    \item\label{item:h_very_weak} %
    For each test function $\xi\in \mC_0^2(\Omega)$ we have $\int_\Omega (h\cdot \Ll_\vphi \xi)\,dxdy = 0$.
    
    \vskip 2pt
    
    \item\label{item:h_W12} Assume additionally that $h$ is differentiable almost everywhere and that the gradient of $h$ is locally square integrable. Then, $h^\ast$ is an $A$-harmonic conjugate of $h$.
  \end{enumerate}
\end{lemma}  
\begin{proof}
  Let $\xi\in\mC_0^1(\Omega)$ be a continuously differentiable real-valued test function. Integrating over level lines of $\xi$ we obtain the identity
\begin{equation}
  \label{eq:Mcon1}
\int_\Omega(h\,d\vpsi+ih^\ast\,dw)\wedge d\xi =0.
\end{equation}
Taking the real and the imaginary parts separately and plugging $\vpsi=\vphi_x+i\vphi_y$, $w=x+iy$ in, this can be rewritten as
\[
\int_\Omega(h(\vphi_{xx}\xi_y-\vphi_{xy}\xi_x)+h^\ast \xi_x)\,dxdy\ =\ 0\ =\ \int_\Omega(h(\vphi_{xy}\xi_y-\vphi_{yy}\xi_x)+h^\ast \xi_y)\,dx dy\,.
\]
Assume now that $\xi$ is twice differentiable. Using the first equation with $\xi$ replaced by $\xi_y$, the second with $\xi$ replaced by $\xi_x$, and subtracting the two, one sees that
\begin{equation}
\label{eq:h-weak_sol}
\int (h\cdot\Ll_\vphi \xi)\,dxdy\ =\ \int_\Omega (h\cdot (\vphi_{xx}\xi_{yy}-2\vphi_{xy}\xi_{xy}+\vphi_{yy}\xi_{xx}))\,dxdy =0
\end{equation}
for each real-valued test function $\xi\in\mC_0^2(\Omega)$, as required in~\ref{item:h_very_weak}.

To prove~\ref{item:h_W12}, assume first that $h$ and $\vpsi$ are smooth. In this case, for each $\xi\in\mC_0^1(\Omega)$ we can use integration by parts and write
\begin{equation}
  \label{eq:Mcon2}
  \int_\Omega \xi dh\wedge d\vpsi= -\int_\Omega hd\vpsi\wedge d\xi.
\end{equation}
Both sides of this equation are continuous with respect to the norm $\|h\|_{\mL^2(U)} + \|\nabla h\|_{\mL^2(U)} + \|\nabla\vpsi\|_{\mL^\infty(U)}$, where $U\subset \Omega$ is any open set containing $\supp \xi$. Thus,~\eqref{eq:Mcon2} is actually valid even when $\vpsi$ is only Lipschitz and $\nabla h$ is locally square integrable, as in this case we can approximate $\vpsi$ and $h$ by smooth functions and use the equality~\eqref{eq:Mcon2} for them. It now follows from~\eqref{eq:Mcon1} that
\begin{equation}
  \label{eq:Mcon3}
  \int_\Omega \xi dh\wedge h\vpsi = i \int_\Omega h^\ast dw\wedge d\xi.
\end{equation}
A direct calculation using~\eqref{eq:A_no_drift} shows that an $A$-harmonic conjugate $\widetilde h^\ast$ of $h$ satisfies the same equation. Therefore, $h^\ast - \widetilde h^\ast$ is locally constant, hence $h^\ast$ is also an $A$-harmonic conjugate.
\end{proof}

One can ask if the property of $h$ being a `very weak', or distributional, solution of the equation $\Ll_\vphi h = 0$ as stated in item~(i) of Lemma~\ref{lemma:Morera_condition} implies that $h$ is differentiable. In the next lemma we give a simple proof of this fact provided that~$\varphi$ is smooth enough.

\begin{lemma}
  \label{lemma:very_weak=weak}
  Assume that $\vphi\in \mC^3(\Omega)$ and $h\in \mC(\Omega)$ is such that for each $\xi\in \mC^2_0(\Omega)$ we have $\int_\Omega (h\cdot \Ll_\vphi \xi)\,dxdy = 0$. Then, $h\in \mC^2(\Omega)$ and $\Ll_\vphi h = 0$.
\end{lemma}
\begin{proof}
  Consider a ball $B=B(p,r)$ such that $\overline{B}\subset \Omega$ and let $h_0$ be the solution of the equation $\Ll_\vphi h_0= 0$ in $B$ with continuous boundary values $h_0=h$ on $\partial B$. Since we assume that $\vphi\in \mC^3$, the matrix $A$ has $\mC^1$-smooth coefficients. Therefore, this strong solution exists and is $\mC^2$-smooth inside~$B$. Thus, it remains to prove that $h=h_0$ in $B$ or that $\int_B ((h-h_0)\cdot \phi)\, dxdy =0$ for all $\phi\in \mC_0^\infty(B)$.
 
  To this end, note that the difference $h-h_0$ is continuous in $\overline{B}$, has zero boundary values on $\partial B$ and satisfies~\eqref{eq:h-weak_sol} for each twice differentiable compactly supported test function $\xi$ in $B$. Now let $\xi$ be the (strong) solution of the nonhomogeneous equation $\Ll_\vphi \xi=\phi$ in $B$ with zero boundary values. As $\phi \in \mC_0^\infty(B)$, classical Schauder estimates~\cite[Chapter~6]{Gilbarg-Trudinger} imply that $\xi$ is $\mC^2$-smooth inside $B$ and has uniformly bounded gradient in $B$ including near $\partial B$. Finally, for sufficiently small $\eps>0$ denote $\xi_\eps=\eta_\eps \xi$, where $\eta_\eps:B\to[0,1]$ is a smooth function such that $\eta_\eps=1$ in $B(p,(1-2\eps)r)$, $\eta_\eps=0$ outside $B(p,(1-\eps)r)$, and $|D^2\eta_\eps|=O(\eps^{-2})$. Since $\xi_\eps\in \mC_0^2(B)$, by applying~\eqref{eq:h-weak_sol} to $h-h_0$ and $\xi_\eps$ we see that
  \begin{equation}
  \label{eq:int(h-h0)g}
  \int_B ((h-h_0)\cdot \phi)\,dxdy\ =\ \int_B ((h-h_0)\cdot \Ll_\vphi \xi)\,dxdy \ =\ \int_B ((h-h_0)\cdot \Ll_\vphi[(1-\eta_\eps)\xi])\,dxdy\,.
  \end{equation}
  Note that $[\Ll_\vphi(1-\eta_\eps)\xi](w)\ne 0$ only if $\eps\le \dist(w,\partial B)\le 2\eps$. For such $w$ we have $[\Ll_\vphi \xi](w)=0$ and $\xi(w)=O(\varepsilon)$ as the gradient of $\xi$ is uniformly bounded in $\overline{B}$. Hence, $\Ll_\vphi[(1-\eta_\eps)\xi](w)=O(\eps^{-1})$, which allows us to estimate the integral~\eqref{eq:int(h-h0)g} by the maximum of $|h-h_0|$ in the $\eps$-neighborhood of $\partial B$. Since $h-h_0$ is continuous in $\overline{B}$, taking the limit as $\eps\to 0$ completes the proof of $h=h_0$.
\end{proof}

Classically (see~\cite[Section~6]{littman-stampacchia-weinberger-1963-regular-points} or~\cite{taylor-kim-brown-2013-green-function} and references therein), the operator $\Ll_\vphi$ admits Green's function in a domain $\Omega$, that is, a function $G_\Omega(w_1,w_2)$ that solves the equation
\[
  \Ll_\vphi G_\Omega(\cdot, w_2) = \delta_{w_2}
\]
for each $w_2\in \Omega$ and vanishes along the boundary of $\Omega$. The standard theory of elliptic equations implies that~$G_\Omega$ is symmetric, positive in $(\Omega\times\Omega)\smm\diag$, and continuous in $(\overline\Omega\times\overline\Omega)\smm\diag$. Moreover, each non-negative continuous function $\wtd G(\cdot, w_2)$ in $\overline\Omega\smm\{ w_2\}$ that vanishes on~$\partial\Omega$ and satisfies the equation $\Ll_\vphi\wtd G(\cdot,w_2) = 0$ in $\Omega\smm\{w_2\}$ coincides with $G_\Omega(\cdot,w_2)$ up to a multiplicative constant. Indeed, a standard application of Harnack inequality (see e.g.~\cite[Section~7]{littman-stampacchia-weinberger-1963-regular-points}) guarantees that $\wtd G(w_1,w_2)$ has the same growth as $G_\Omega(w_1,w_2)$ as $w_1\to w_2$, which implies the assertion due to a classical result from~\cite{serrin-weinberger-1966-isolated-singularities}. The unknown multiplicative constant can be fixed by computing the monodromy of the $A$-harmonic conjugate:

\begin{lemma}
  \label{lemma:Green_function_characterization}
  Let $G_\Omega$ be the Green function of the operator $\Ll_\vphi$ in~$\Omega$. Fix $w_2\in \Omega$ and let $G_\Omega^\ast(\,\cdot\,,w_2)$ be an $A$-harmonic conjugate of $G_\Omega(\,\cdot\,,w_2)$. Then, the function $G_\Omega^\ast(\,\cdot\,,w_2)$ is additively multivalued and its monodromy around $w_2$ in the counterclockwise direction equals $-1$.
\end{lemma}
\begin{proof}
  Using standard continuity arguments (e.g., see~\cite[Section~5]{littman-stampacchia-weinberger-1963-regular-points}) it is enough to consider smooth $\vphi$, in which case $G_\Omega$ is smooth too due to Schauder estimates. Fix a smooth function $\xi:\RR\to\RR$ with a small support that is constant in a vicinity of $0$. Let $m$ be the monodromy of $G_\Omega^\ast(\cdot,w_2)$. Using the definition of an $A$-harmonic conjugate and integrating along concentric circles around $w_2$ we see that
  \[
    -m \xi(0) = \int_\Omega \nabla G_\Omega^\ast(w,w_2)\cdot \ast\nabla\xi(|w-w_0|)\,d^2w = \int_\Omega A\nabla G_\Omega(w,w_2)\cdot \nabla\xi(|w-w_0|)\,d^2w = \xi(0),
  \]
  which gives $m = -1$.
\end{proof}

\section{Convergence of discrete harmonic functions on $\Gamma_\delta$ as $\delta\to 0$}
\label{sec:Convergence of harmonic functions in C0 topology}

The goal of this section is to prove Theorem~\ref{thma:harmonic_functions_convergence} and Theorem~\ref{thma:Green_fct_convergence}. Our strategy can be described as follows. We begin by observing that property~\ref{prty:LIP}, Lemma~\ref{lemma:Lip_follows_from_our_assumptions}, and Proposition~\ref{prop:RW_ellipticity} yield the precompactness of the family $(H_\delta)_{\delta\to 0}$ in the $\mC^0$ topology (and, in fact, in $\mC^\alpha$ for a small $\alpha>0$) and the fact that all subsequential limtis of $H_\delta$ have correct boundary conditions. Further, discrete Dirichlet energy of $H_\delta$ can be bounded on compacts due to a discrete Caccioppoli estimate which, in its turn, allows to estimate \emph{harmonic conjugate} functions $H^\ast_\delta$ and to conclude that the family $(H^\ast_\delta)_{\delta\to 0}$ is precompact in the $\mC^0$ topology as well. Passing to a subsequence and applying Lemma~\ref{lemma:closed_form} we obtain a pair of continuous functions $h,h^\ast$ such that the 1-form $h\,d\vpsi + ih^\ast\,dw$ is closed. This allows us to recover the coefficients of the elliptic equation that $h$ satisfies by using results from Section~\ref{sec:Elliptic equations and closed 1-forms in 2D}.

\subsection{Oscillations estimate via the Dirichlet energy}
\label{subsec:Dirichlet energy estimates}
In this section we recall a well-known equicontinuity estimate for discrete harmonic functions via their Dirichlet energies; see Lemma~\ref{lemma:harmonic_conjugate_bound} and Remark~\ref{rem:Dirichlet_bound_log_continuity} below. In its simplest form, this estimate dates back to Lusternik~\cite{Lusternik1926}; see also~\cite{skopenkov-Dirichlet,gurel-gurevich-jerison-nachmias} and references therein.

In what follows we assume that $\Gamma_\delta$ is a family of graphs embedded harmonically that satisfy property~\ref{prty:LIP}.
Given $U\subset \CC$, we denote by $U_\delta$ the set of vertices of $\Gamma_\delta$ that belong to $U$ and by $\partial U_\delta$ the set of vertices of $\Gamma_\delta$ that do not belong to $U$ but have neighbors in $U_\delta$. We put $\overline U_\delta = U_\delta\cup \partial U_\delta$ and denote by $E(U_\delta)$ the set of edges of $\Gamma_\delta$ with both ends in $U_\delta$. Also, we denote by $E(U_\delta^\ast)$ the set of edges of $\Gamma_\delta^\ast$ that are dual to $E(U_\delta)$ and by $U_\delta^\ast$ the set of all vertices of $\Gamma_\delta^\ast$ incident to these edges.

\begin{defin}
  \label{defin:Dirichlet_energy}
  Let $\delta>0$ and $U\subset \CC$. The Dirichlet energy of a function $H:U_\delta\to\RR$ is given by 
  \[
    \Ee^\delta_U(H)\ \ =\!\! \sum_{v_1v_2\in E(U_\delta)} c_{v_1v_2}(H(v_1) - H(v_2))^2.
  \]
  Similarly, for a function $H^\ast:U_\delta^\ast\to\RR$ we define its Dirichlet energy by
  \[
    \Ee^{\delta, \ast}_U(H^\ast)=\sum_{v^\ast_1 v^\ast_2\in E(U_\delta^\ast)}c_{v_1v_2}^{-1}(H^\ast(v^\ast_1) - H^\ast(v^\ast_2))^2.
  \]
\end{defin}

\begin{lemma}
  \label{lemma:C1_functions_DE_bound} 
  Assume that $U\subset \CC$ is bounded and that there exists $r_0>0$ such that the {$r_0$-neighborhood} of $U$ is covered by faces of $\Gamma_\delta$ for all small enough $\delta>0$. Let $g$ be a continuously differentiable function defined in a $\delta$-neighborhood of $U$ (and hence on all $U_\delta$). Then, %
  \[
    \textstyle \Ee_U^\delta(g)\ \leq\ \|\nabla g\|_\infty^2\cdot \sum_{v_1v_2\in E(U_\delta)} c_{v_1v_2}|v_1-v_2|^2 .
  \]
\end{lemma}
\begin{proof} This follows from the trivial estimate $|g(v_2)-g(v_1)|\le\|\nabla g\|_\infty\cdot|v_2-v_1|$.
\end{proof}

Given two subsets $A,B\subset U_\delta$, we denote by $\EL_U^\delta(A\leftrightarrow B)$ the \emph{(edge) extremal length} of paths connecting $A$ and $B$ inside $U_\delta$; this quantity is also known as the \emph{effective resistance} between $A$ and $B$ in the electrical network $U_\delta$. Similarly, given $A^\ast,B^\ast\subset U_\delta^\ast$, we denote by $\EL_U^{\delta,\ast}(A^\ast\leftrightarrow B^\ast)$ the extremal length of paths connecting $A^\ast$ and $B^\ast$ inside $U_\delta^\ast$. We address the reader to~\cite[Chapter~2]{lyons-peres-book} or to \cite[Section~6]{chelkak-robust},~\cite[Section~2.2]{binder2024orthodiagonal} and references therein for the definition and basic properties of the discrete extremal length.

For $K\subset \CC$ and a function $H^\ast: K^\ast_\delta\to \RR$ define
\[
  \osc_K H^\ast=\max_{K_\delta^\ast} H^\ast - \min_{K_\delta^\ast} H^\ast.
\]

\begin{lemma}
  \label{lemma:harmonic_conjugate_bound}
  Let $U\subset \CC$ be an open simply connected set and assume that $U$ is covered by faces of $\Gamma_\delta$ for all small enough $\delta>0$. For each compact subset $K\subset U$ there exist constants $C,\delta_0>0$ such that for each $\delta\leq \delta_0$ and each discrete harmonic function $H:\overline U_\delta\to \RR$ we have
  \begin{equation}
    \label{eq:harmonic_conjugate_bound}
    \osc_K H^\ast \leq C \sqrt{\Ee^\delta_U(H)},
  \end{equation}
where $H^\ast$ is a harmonic conjugate to $H$. (Note that $H^\ast$ is uniquely defined, up to an additive constant, in the bulk of $U$ provided that $\delta$ is sufficiently small.)
\end{lemma}
\begin{proof}
  It is enough to prove that for each $p\in U$ there exist $C,r>0$ such that~\eqref{eq:harmonic_conjugate_bound} holds for $K = \overline B(p,r)$ and small enough $\delta$. Let $r = \frac{1}{2}\dist(p,\partial U)$ and $A(p,r,2r) = B(p,2r)\smm \overline B(p,r)$. 
  
  Using the maximum principle for the function $H^\ast$ it is easy to show the existence of two paths $\gamma_M^\ast$ and $\gamma_m^\ast$ connecting the boundary components of the discrete annulus $A(p,r,2r)_\delta^\ast$ such that
  \[
    H^\ast\vert_{\gamma_M^\ast} \geq \max_{B(p,r)_\delta^\ast}H^\ast,\qquad H^\ast\vert_{\gamma_m^\ast} \leq \min_{B(p,r)_\delta^\ast}H^\ast.
  \]
  Basic properties of the extremal length/effective resistance immediately imply that
  \[%
    \left( \osc_{B(p,r)}H^\ast \right)^2 
    \leq\ \Ee_{B(p,2r)}^\delta(H)\cdot \EL^{\delta,\ast}_{A(p,r,2r)}(\gamma_M^\ast \leftrightarrow \gamma_m^\ast),
  \]%
  where we also used the fact that the Dirichlet energies of the harmonic function $H$ and of its harmonic conjugate $H^\ast$ are equal.
  We can now bound $\EL^{\delta,\ast}_{A(p,r,2r)}(\gamma_M^\ast \leftrightarrow \gamma_m^\ast)$ from above by the extremal length of circuits separating the two boundary components of the annulus. Classically~\cite[Theorem~1]{ford1956maximal}, the latter equals the inverse of the extremal length $\EL^\delta_{A(p,r,2r)}(\partial B(p,r)\leftrightarrow \partial B(p,2r))$, where $\partial B(p,r)$ denotes the set of vertices in $A(p,r,2r)_\delta$ incident to vertices from $\overline B(p,r)_\delta$, and the same for $\partial B(p,2r)$. Therefore,
  \begin{equation*} %
    \left( \osc_{B(p,r)}H^\ast \right)^2 \leq\ \Ee_{B(p,2r)}^\delta(H)\cdot (\EL^\delta_{A(p,r,2r)}(\partial B(p,r)\leftrightarrow \partial B(p,2r))^{-1}.
  \end{equation*}
  In order to estimate $\EL^\delta_{A(p,r,2r)}(\partial B(p,r)\leftrightarrow \partial B(p,2r))$ from below we can consider the metric given by the discrete gradient of the function $g_p(v) = \log|v - p|$. Namely,
  \[%
    \EL^\delta_{A(p,r,2r)}(\partial B(p,r)\leftrightarrow \partial B(p,2r))%
    \ \geq\ \left(\min_{\partial B(p,2r)} g_p - \max_{\partial B(p,r)} g_p\right)^{\!2}\cdot (\Ee^\delta_{A(p,r,2r)}(g_p))^{-1}.
  \]%
  This allows us to conclude that
  \begin{equation}
  \label{eq:oscillations_bound}
  \left(\osc_{B(p,r)}H^\ast \right)^2\ \leq\ (\log 2)^{-2}\cdot \Ee^\delta_{A(p,r,2r)}(g_p)\cdot \Ee_{B(p,2r)}^\delta(H)\,.
  \end{equation}
  Due to Lemma~\ref{lemma:C1_functions_DE_bound} and Lemma~\ref{lemma:white_area_is_bdd}, $\Ee^\delta_{A(p,r,2r)}(g_p)$ is bounded from above by a constant that depends on the constant in property~\ref{prty:LIP} only (provided that $\delta$ is small enough). The claim easily follows. 
\end{proof}

\begin{rem}
  \label{rem:osc_of_primal_H}
  Similar arguments %
  can be used to estimate the oscillations of a harmonic function $H$ on the primal graph $\Gamma_\delta$. A subtle difference is that on the last step we need to estimate the \emph{dual} extremal length $\EL^{\delta,\ast}_{A(p,r,2r)}(\partial B(p,r)\leftrightarrow \partial B(p,2r))$, which is comparable to the conformal modulus of the annulus $\vPsi_\delta(A(p,r,2r))$. This conformal modulus is uniformly bounded due to Lemma~\ref{lemma:Lip_follows_from_our_assumptions}.
\end{rem}

\begin{rem}
\label{rem:Dirichlet_bound_log_continuity}
  A simple modification of the arguments from Lemma~\ref{lemma:harmonic_conjugate_bound} implies that for each compact $K\subset U$ and a discrete harmonic function $H$ on $U_\delta$ one has
  \begin{equation}
    \label{eq:log_Holder_via_capacity}
    \omega_H(t,K) \leq C|\log t|^{-1/2}\cdot \sqrt{\Ee^\delta_U(H)} \ \ \text{for all}\ \ t\geq C\delta\,,
  \end{equation}
    where $\omega_H(t,K)$ is the modulus of continuity of $H$ on $K$ and 
    the constant $C$ depends on $K,U$, and the constants in property~\ref{prty:LIP} only. 
    In fact, Proposition~\ref{prop:RW_ellipticity} gives a stronger H\"older-type estimate for $\omega_H(t,K)$ once we know that the oscillations of $H$ are bounded. As explained in the next remark, one can also give a self-contained proof of Proposition~\ref{prop:RW_ellipticity} that does not rely upon a more general result provided by~\cite[Proposition~6.4]{CLR1}.
\end{rem}

\begin{rem}
  \label{rem:ellipticity}
  Note that the arguments used above to derive the inequality~\eqref{eq:log_Holder_via_capacity} can be applied to any (not necessarily harmonic) function $H$ that satisfies the following property: for each $v\in U_\delta$ there exist two paths $\gamma_+,\gamma_-$ connecting $v$ to $\partial U_\delta$ such that $H(v_+)\geq H(v) - \cst\cdot\delta$ for all $v_+\in \gamma_+$ and $H(v_-)\leq H(v) + \cst\cdot\delta$ for all $v_-\in \gamma_-$. This observation can be used to give an independent proof of the ellipticity estimate from Proposition~\ref{prop:RW_ellipticity} as follows: 
  
  Pick $v_0\in U_\delta$ and consider the square $Q = v_0 + [-L\delta, L\delta]^2$, where $L$ is a sufficiently large constant. Assume that a non-negative smooth function $f$ equals 1 at the middle of the right side of $Q$ (that is, at the point $v_0 + L\delta$) and vanishes in an $O(\delta)$-neighborhood of the other three sides. Let a function $H$ be harmonic in~$Q_\delta$ and equal to $f$ outside~$Q$. The Dirichlet energy $\Ee_{B(v_0, 2L\delta)}^\delta(H) \leq \Ee_{B(v_0, 2L\delta)}^\delta(f)$ is bounded by an absolute constant and one can choose~$f$ so that $H$ satisfies the aforementioned property on the existence of paths~$\gamma_\pm$. %
  Then, the estimate~\eqref{eq:log_Holder_via_capacity} implies that $H$ is bounded away from zero in a small vicinity of the point $v_0 + L\delta$. In its turn, this gives a uniform lower bound on the probability that a random walk started in this vicinity exits $Q$ through its right side. Finally, one can show that such a weaker version form of Lemma~\ref{lemma:crossing_estimates} is ultimately equivalent to Proposition~\ref{prop:RW_ellipticity}.
\end{rem}

\subsection{Discrete Caccioppoli estimate}
\label{subsec:Discrete Caccioppoli estimate}

Caccioppoli's theorem is a classical result in the theory of elliptic equations asserting that the $\mL^2$ norm over a ball $B(p,r)$ of the gradient of a solution to an elliptic equation can be controlled via the $\mL^2$ norm of the solution itself over a larger ball $B(p,2r)$. This admits a simple analogue for discrete harmonic functions.

\begin{prop}
  \label{prop:discrete_Caccioppoli}
  Let $\Gamma_\delta$ be a weighted planar graph embedded harmonically into the complex plane. 
  Let $r>0$ and $p\in \CC$ be such that the disc $B(p,2r)$ is covered by faces of $\Gamma_\delta$. Then, for each harmonic function $H:\overline {B(p,2r)}_\delta\to \RR$ the following estimate holds: %
  \[
    \Ee_{B(p,r)}^\delta(H)\ \leq\ \frac{2}{r^2}\ \times\!\!\!\!\!\! \sum_{v_1v_2\in E({B(p,2r)}_\delta)} c_{v_1v_2}|v_1 - v_2|^2(H(v_1) + H(v_2))^2.
  \]
\end{prop}
\begin{proof} Denote $B_\delta=B(p,2r)_\delta$ for shortness.
  Let $\phi: \CC\to [0,1]$ be a Lipschitz function that equals 1 over $B(p,r)$, vanishes at the boundary of $B(p,2r)$ and outside this disk, and such that $|\nabla\phi|\le r^{-1}$. As $H$ is discrete harmonic and $\phi$ vanishes outside $B(p,2r)$, the discrete integration by parts gives
  \[ %
    \sum_{v_1v_2\in E(B_\delta)} c_{v_1v_2}(H(v_1) - H(v_2))(\phi(v_1)^2H(v_1) - \phi(v_2)^2H(v_2)) = 0.
  \] %
  Simple algebraic manipulations allow one to rewrite this identity as
  \[%
    \sum_{v_1v_2\in E(B_\delta)} c_{v_1v_2}(H(v_1) - H(v_2))^2(\phi(v_1)^2 + \phi(v_2)^2)\ =
    \!\!\! \sum_{v_1v_2\in E(B_\delta)} c_{v_1v_2}(H(v_1)^2-H(v_2)^2)%
    (\phi(v_2)^2 - \phi(v_1)^2).
  \]%
  Applying Cauchy--Schwartz's inequality one sees that the square of the right-hand side is bounded from above by
  \[%
  \sum_{v_1v_2\in E(B_\delta)} c_{v_1v_2}(H(v_1) - H(v_2))^2(\phi(v_1)+\phi(v_2))^2
  \ \ \times\!\!\sum_{v_1v_2\in E(B_\delta)} c_{v_1v_2}(H(v_1) + H(v_2))^2(\phi(v_2) - \phi(v_1))^2.
  \]%
  As $(\phi(v_1)+\phi(v_2))^2\le 2(\phi(v_1)^2+\phi(v_2)^2)$, this gives the estimate
  \[ 
    \sum_{v_1v_2\in E(B_\delta)} c_{v_1v_2}(H(v_1) - H(v_2))^2(\phi(v_1)^2 + \phi(v_2)^2)\ \leq\
    2\!\!\!\!\sum_{v_1v_2\in E(B_\delta)} c_{v_1v_2}(H(v_1) + H(v_2))^2(\phi(v_2) - \phi(v_1))^2\,.
  \]
  Using the fact that $\phi=1$ on $B(p,r)$ we conclude that
  \begin{align*} 
    \Ee_{B(p,r)}^\delta(H)\ &\le \!\!\sum_{v_1v_2\in E(B_\delta)} c_{v_1v_2}(H(v_1) - H(v_2))^2(\phi(v_1)^2+\phi(v_2)^2)\\
    &\le\  2\!\!\!\!\sum_{v_1v_2\in E(B_\delta)} c_{v_1v_2}(H(v_1)+H(v_2))^2(\phi(v_2) - \phi(v_1))^2\\
    &\le\ 2\|\nabla \phi\|_\infty^2\ \times\!\!\!\!\!\!\sum_{v_1v_2\in E({B(p,2r)}_\delta)} c_{v_1v_2}|v_1 - v_2|^2(H(v_1) + H(v_2))^2.
  \end{align*} 
  The claim follows.
\end{proof}

\begin{cor}
  \label{cor:Caccioppoli}
  Let $(\Gamma_\delta)_{\delta\to 0}$ be a family of weighted planar graphs embedded harmonically into $\CC$ that satisfy property~\ref{prty:LIP}. Assume that $r>0$ and $p\in \CC$ are such that the disc $B(p,3r)$ is covered by the faces of $\Gamma_\delta$ for each sufficiently small $\delta>0$. There exists a constant $C=C(\kappa)$ and $\delta_0>0$ such that for each $\delta\leq \delta_0$ and each harmonic function $H:\overline {B(p,2r)}_\delta\to \RR$ we have
  \[
    \Ee_{B(p,r)}^\delta(H)\ \leq\ Cr^{-2}\max\nolimits_{v\in \overline{B(p,2r)}_\delta} (H(v))^2.
  \]
\end{cor}
\begin{proof}
  This follows from Lemma~\ref{lemma:white_area_is_bdd} and Proposition~\ref{prop:discrete_Caccioppoli}.
\end{proof}

\subsection{Proofs of Theorem~\ref{thma:harmonic_functions_convergence} and Theorem~\ref{thma:Green_fct_convergence} under the assumption $\vphi\in \mC^3$}
\label{subsec:proof of Theorem12 vphi smooth}

We begin with the proof of Theorem~\ref{thma:harmonic_functions_convergence} and then discuss the only additional ingredient -- a priori uniform bound on the values of the Green function -- that is needed to prove Theorem~\ref{thma:Green_fct_convergence} along the same lines.

\begin{proof}[Proof of Theorem~\ref{thma:harmonic_functions_convergence} under the assumption $\vphi\in \mC^3$]
Note that the functions $H_\delta$ are uniformly bounded by maximum principle. Corollary~\ref{cor:Caccioppoli} implies that for each compact $K\subset \Omega$ there is a constant $C = C(K)$ such that
\begin{equation}
  \label{eq:Ee(H)_is_bdd}
  \Ee^\delta_K(H_\delta) \leq C(K).
\end{equation}
Lemma~\ref{lemma:harmonic_conjugate_bound} implies that harmonic conjugate functions $H_\delta^\ast$ of $H_\delta$ can be chosen to be uniformly bounded on compacts of $\Omega$. Combining property~\ref{prty:LIP}, Lemma~\ref{lemma:Lip_follows_from_our_assumptions} and Proposition~\ref{prop:RW_ellipticity} we conclude that the family $(H_\delta,H_\delta^\ast)_{\delta\to 0}$ is precompact in the topology of uniform convergence on every compact in $\Omega$. Let $(h,h^\ast)$ be a subsequential limit of this family. Denote $f = h + ih^\ast$ and let $z=\frac12(w+\vpsi)$, $\vartheta=\frac12(\overline{\vpsi}-\overline{w})$. By Lemma~\ref{lemma:closed_form} we know that the 1-form 
\[
f\,dz + \bar f\,d\bar \vtheta = h\,d\vpsi+ih^\ast\,dw
\] 
is closed `in the weak sense': namely, we have $\int_\gamma (f\,dz + \bar f\,d\bar \vtheta) = 0$ for each contractible piecewise smooth loop $\gamma\subset \Omega$. Applying Lemma~\ref{lemma:Morera_condition} and Lemma~\ref{lemma:very_weak=weak} (here we use that $\vphi\in \mC^3$) we conclude that $h\in \mC^2(\Omega)$, we have $\Ll_\vphi h = 0$, and $h^\ast$ is an $A$-harmonic conjugate of $h$.

Using Lemma~\ref{lemma:crossing_estimates} it is easy to see that $h$ is continuous up to $\overline \Omega$, we have $h\vert_{\partial \Omega} = g$, and that $H_\delta$ converge to $h$ uniformly in $\overline \Omega$ along the corresponding subsequence. It follows that $h$ solves the Dirichlet problem for $\Ll_\vphi$ with boundary conditions $g$, thus $h$ satisfies the assumptions of the theorem. Using the uniqueness of the solution of this Dirichlet problem we conclude that the full sequence $H_\delta$ converges to $h$ as $\delta\to 0$.
\end{proof}

Among the arguments used in the proof of Theorem~\ref{thma:harmonic_functions_convergence} given above, the only missing ingredient in the setup of Theorem~\ref{thma:Green_fct_convergence} is an a priori uniform bound on the values of discrete harmonic functions under consideration, from which the proof starts. The next lemma provides this missing ingredient. 

\begin{lemma}
  \label{lemma:Green_function_bounded}
  In the setup of Theorem~\ref{thma:Green_fct_convergence}, for each compact $K\subset (\overline\Omega\times \overline\Omega) \smm \diag$ there exists a constant $C>0$ that depends on $K$ and the constants in property~\ref{prty:LIP} only such that for each sufficiently small $\delta$ we have
  \[
    \max_{(v_1,v_2)\in K} G_{\Omega_\delta}(v_1,v_2)\leq C.
  \]
\end{lemma}
\begin{proof}
   Without loss of generality, we can assume that $\dist(v_2,\partial\Omega)\geq C\delta$ for some $C>0$ big enough: the desired estimate in the other case follows from the monotonicity with respect to the domain. Below we use the symbol $O(\cdot)$ to denote an estimate with a constant that depends on the constants from property~\ref{prty:LIP} only.

  Let $X_t$ be the continuous time random walk on $\Gamma_\delta$ with jump rates from~$v$ to~$v'$ proportional to~$c_{vv'}$ and parameterized such that $|X_t|^2 - t$ is a local martingale. Recall that
  \[
    \mu_\delta(v) = \sum_{v_1\sim v} c_{vv'}|v-v'|^2
  \]
  is an invariant (in the bulk) measure for $X_t$; recall that we include $\mu_\delta(v)$ as the normalizing factor in the definition~\eqref{eq:def_of_Lldelta} of the discrete Laplacian. By Dynkin's formula we have
  \begin{equation}
    \label{eq:G_Dynkin}
    G_{\Omega_\delta}(v_1,v_2) = \frac{1}{\mu_\delta(v_2)} \EE\int_0^{\tau_{\Omega_\delta}} \indic[X_t^{v_1} = v_2]\,dt,
  \end{equation}
  where $X_t^v$ stands for the random walk started at $v$ and $\tau_{\Omega_\delta}$ is the first  exit time from $\Omega_\delta$.

  Let $r = \frac{1}{4}\dist(v_2,\partial\Omega)$ and assume that $|v_1-v_2|\geq 2r$. Using the fact that $|X_t|^2-t$ is a local martingale and crossing estimates from Lemma~\ref{lemma:crossing_estimates} it is easy to see that
  \begin{equation*}
    \EE\int_0^{\tau_{\Omega_\delta}}\indic[X_t^{v_1}\in B(v_2,r)] = O(r^2).
  \end{equation*}
  Due to~\eqref{eq:G_Dynkin}, this gives
  \begin{equation*}
    \sum_{v\in B(v_2,r)_\delta}\mu_\delta(v)G_{\Omega_\delta}(v_1,v) = O(r^2).
  \end{equation*}
  By Harnack inequality (Lemma~\ref{lemma:Harnack}) there exist $c>0$ such that for each $v\in B(v_2,r)$ we have
  \begin{equation*}
    G_{\Omega_\delta}(v_1,v)\geq c\cdot G_{\Omega_\delta}(v_1,v_2).
  \end{equation*}
  From Lemma~\ref{lemma:white_area_is_bdd} we also know that $\sum_{v\in B(p,r)_\delta} \mu_\delta(v)\ge cr^2$ for some $c>0$. Thus,  $G_{\Omega_\delta}(v_1,v_2)=O(1)$. 
We conclude that $G_{\Omega_\delta}(v_1,v_2)$ is uniformly bounded when $|v_1-v_2|\geq \max(C\delta, \frac{1}{2}\dist(v_2,\partial\Omega))$. This implies the estimate over an arbitrary compact $K\subset (\overline \Omega\times\overline \Omega)\smm\diag$ by Harnack inequality.
\end{proof}

\begin{proof}[Proof of Theorem~\ref{thma:Green_fct_convergence} under the assumption $\vphi\in \mC^3$]
Using Lemma~\ref{lemma:Green_function_bounded} and Lemma~\ref{lemma:Holder}, as well as Lemma~\ref{lemma:crossing_estimates} to estimate $G_{\Omega_\delta}$ near the boundary, we conclude that the family $(G_{\Omega_\delta})_{\delta>0}$ is precompact in the topology of uniform convergence on compact subsets of $(\overline\Omega\times\overline\Omega)\smm \diag$. It remains to prove that each subsequential limit coincides with the Green function $G_\Omega$ of the operator $\Ll_\vphi$ in the domain~$\Omega$. To this end, we employ the characterization of $G_\Omega$ discussed before Lemma~\ref{lemma:Green_function_characterization}.

Let $\wtd G$ be such a subsequential limit. Clearly, $\wtd G$ is non-negative, continuous in $(\overline\Omega\times\overline\Omega)\smm \diag$, and equal to zero on $\partial\Omega$. Now fix $w_2\in \Omega$ and denote $G_\delta = G_{\Omega_\delta}(\,\cdot\,,w_2)$ for shortness. Corollary~\ref{cor:Caccioppoli} and Lemma~\ref{lemma:harmonic_conjugate_bound} allow us to define harmonic conjugates $G^\ast_\delta$ so that they are uniformly bounded on compact subsets of the universal cover of $\Omega\smm\{ w_2 \}$. Arguing as in the proof of Theorem~\ref{thma:harmonic_functions_convergence} given above we see that $\wtd G(\,\cdot\,,w_2)\in \mC^2(\Omega\smm\{ w_2 \})$, we have $\Ll_\vphi \wtd G(\,\cdot\,, w_2) = 0$ in $\Omega\smm\{ w_0 \}$, and $G_\delta^\ast(\cdot, w_2)$ converge to an $A$-harmonic conjugate of $\wtd G(\,\cdot\,, w_2)$ uniformly on compact subsets of the universal cover of $\Omega\smm\{ w_0 \}$. In particular, this means that an $A$-harmonic conjugate of $\wtd G(\,\cdot\,, w_2)$ has monodromy $-1$ around $w_2$. Due to Lemma~\ref{lemma:Green_function_characterization}, this characterizes $\wtd G$ as the Green function of $\Ll_\vphi$.
\end{proof}

\subsection{$C^1$ convergence under additional regularity assumption~\ref{assum:Exp-Fat}}
\label{subsec:C1 under ExpFat} The main convergence results of our paper, Theorem~\ref{thma:harmonic_functions_convergence} and Theorem~\ref{thma:Green_fct_convergence}, guarantee the convergence of discrete harmonic functions~$H_\delta$ to solutions of the equation $\Ll_\varphi h=0$ under no local `non-degeneracy' assumptions on the embeddings $\Gamma_\delta$. At the same time, in applications of such results to the dimer model one often needs to prove convergence of the \emph{gradients} of $H_\delta$. (For instance, see~\cite[Section~1.7]{berestycki-lis-qian-free-dimers} for a short discussion of the link between dimer model observables in this setup and discrete Green's functions.)
In this section we strengthen Theorem~\ref{thma:harmonic_functions_convergence} and Theorem~\ref{thma:Green_fct_convergence} to the convergence of gradients of~$H_\delta$ under a  mild additional regularity assumption~\ref{assum:Exp-Fat} on~$\Gamma_\delta$ that originated in~\cite{CLR1}. 
We say that a face of $\Gamma_\delta$ is \emph{$\rho$-fat} if it can be triangulated so that each triangle contains a disc of radius~$\rho$.

\renewcommand{\theassum}{(EXP-FAT)}

\begin{assum}[\protect{cf.~\cite[Assumption~5.9]{CLR1}}] 
  \label{assum:Exp-Fat}
  Assume that a family of harmonic embeddings $(\Gamma_\delta)_{\delta\to 0}$ and an open set $U\subset \CC$ covered by each of~$\Gamma_\delta$ be given. We say that $\Gamma_\delta$ have property $\ExpFat$ in $U$ if there exists a function $\delta'(\delta)$ such that $\delta'(\delta)\to 0$ as $\delta\to 0$ and the following holds:
  \begin{center}
    \noindent   if one removes from $\Gamma_\delta$ all $\delta\exp(-\delta'(\delta)\delta^{-1})$-fat faces, then all vertex-connected\\ components of~$\Gamma_\delta$ that are fully contained in $U$ have diameters less than~$\delta'(\delta)$.
    \end{center}
\end{assum}

Let us begin with a short motivating discussion. Assume that the potential $\varphi$ is smooth enough and that a real-valued function $h$ is a solution of the equation $\Ll_\vphi h = 0$. A direct computation shows that this is equivalent to the fact that the 1-form $\Im[h_w \,d\vphi]$ is closed. Note also that the 1-form $\Re[h_w\,dw]=\frac12 dh$ is always closed. %
These real-valued 1-forms are the imaginary and the real parts of the complex-valued 1-form 
\begin{equation}
  \label{eq:form_on_t-surf_continuous}
  h_w\,dz - h_{\bar w}\,d\theta,\qquad z = \tfrac12(w + \psi),\ \ \theta = \tfrac12(\bar\psi - \bar w),
\end{equation}
which is therefore also closed. Let $\alpha\in\TT$ and note that
\[
h_w\,dz - h_{\bar w}\,d\theta = \bar\alpha \big( \Re[\alpha h_w]\,d\psi_\alpha + i\Im[\alpha h_w]\,dw_\alpha \big),\qquad
    w_\alpha = z + \alpha^2\theta,\ \ \psi_\alpha = z - \alpha^2\theta.
\]
Repeating the proof of Lemma~\ref{lemma:Morera_condition} one can now see that each projection $\Re[\alpha h_w]$ of the gradient~$h_w$ of $h$ satisfies an elliptic PDE (which depends on $\alpha$) in the coordinate $w_\alpha$.

The framework developed in~\cite{CLR1} provides a `discrete version' of the computations made above when a smooth potential~$\vphi$ and a function~$h$ are replaced with the Maxwell--Cremona potential~$\vPhi_\delta$ and a discrete harmonic function~$H_\delta$ on~$\Gamma_\delta$, respectively. Splitting each face of $\Gamma_\delta$ into triangles and extending $H_\delta$ inside these triangles linearly we can view $H_\delta$ as a piecewise linear function and its gradient $F_\delta:=\partial_w H_\delta$ as a piecewise constant function. (Note the advantage of the discrete setup: $F_\delta$ is well-defined pointwise for each discrete harmonic function $H_\delta$.)
The discrete analog of the 1-form~\eqref{eq:form_on_t-surf_continuous} is then the 1-form $F_\delta d\Tt_\delta - \overline{F}_\delta d\Oo_\delta$ on the t-surface $\Theta_\delta$ associated with $\Gamma_\delta$. 
(See~\cite[Section~4.2 and~Section~5]{CLR1} for an extension of this 1-form from black faces~$\Theta_\delta$ to the whole t-surface.)
As above, this gives a certain discrete harmonicity property for each of the functions~$\Re[\alpha F_\delta]$, $\alpha\in\TT$, in the coordinate $\Tt_\delta+\alpha^2\Oo_\delta$; see~\cite[Section~4.3]{CLR1} for details. Importantly, if $\Gamma_\delta$ has property~\ref{prty:CONV}/\ref{prty:LIP}, with \emph{no} additional assumptions on~$\Gamma_\delta$, 
then the corresponding random walks on T-graphs $\Tt_\delta+\alpha^2\Oo_\delta$ are uniformly elliptic starting from the scale~$\delta$. This leads to the following:

\begin{prop}
  \label{prop:Hw_max_Holder_alternative} 
  Let a harmonic embedding $\Gamma_\delta$ cover an open set~$U\subset\CC$, and $H_\delta:\overline{U}_\delta\to\RR$ be a discrete harmonic function. Triangulate faces of $\Gamma_\delta$ arbitrarily and extend $H_\delta$ to $U$ by linearity.\\[2pt]
  (i) The gradient $\nabla H_\delta$ of $H_\delta$ satisfies the maximum principle: $\max_U |\nabla H_\delta| = \max_{\partial U}|\nabla H_\delta|$, where we take the maximum of $|\nabla H_\delta|$ over all incident faces if $\partial U$ passes through a vertex of $\Gamma_\delta$.\\[2pt]
  (ii) Assume in addition that $\Gamma_\delta$ has property~\ref{prty:LIP}. There exists constants $C,\beta>0$ (depending on constants in~\ref{prty:LIP} only) such that for each $R\geq r\geq C\delta$ and $w\in U$ such that $B(w,R)\subset U$ we have
      \[
        \osc_{B(w,r)} \nabla H_\delta\ \leq\ C(r/R)^\beta \osc_{B(w,R)} \nabla H_\delta\,,
      \]
  where $\osc_B F=\sup_{w_1,w_2\in B}|F(w_1)-F(w_2)|$. The exponent $\beta$ depends only on the constant~$\lambda$ in~\eqref{eq:Conv}.\\[2pt]
  (iii) In the same setup, there exist constants $\beta_0,C_0>0$ depending on the constants in property~\ref{prty:LIP} only such that the following holds. Let $r\geq C_0\delta$ and $w\in U$ be such that $B(w,r)\subset U$. Then,
  \begin{align*}
    \text{either } &\max\nolimits_{B(v,\frac12r)} |\nabla H_\delta|\ \leq\ C_0r^{-1}\osc_{B(v,r)}H_\delta\,,\\
    \text{or } &\max\nolimits_{B(v,\frac34r)} |\nabla H_\delta|\ \geq\ \exp(\beta_0r\delta^{-1})C_0r^{-1}\osc_{B(w,r)}H_\delta\,.
  \end{align*}
\end{prop} 
\begin{proof}
Note that the gradient of $H_\delta$ can be seen as a t-black-holomorphic function on the t-embedding $\Tt_\delta=\frac12(\Tt_\delta+\Oo_\delta)$; see the discussion at the end of Section~\ref{subsec:T-holomorphic functions} and~\cite[Proposition~8.2(ii)]{CLR1}. Recall also that the mapping from $\Gamma_\delta$ to $\Tt_\delta$ is bi-Lipschitz on convex sets starting from the scale~$\delta$. Item~(i) follows from~\cite[Proposition~4.17 and~Proposition~5.7]{CLR1}. 
  Items~(ii) and~(iii) follow from~\cite[Proposition~6.13 and~Theorem~6.17]{CLR1}.
\end{proof}

Given the alternative provided by item~(iii) in Proposition~\ref{prop:Hw_max_Holder_alternative}, we can now benefit from the fact that the gradient of a discrete harmonic function~$H_\delta$ on a $\rho$-fat face of $\Gamma_\delta$ cannot be bigger than $\rho^{-1}\max|H_\delta|$. Together with the maximal principle for $\nabla H_\delta$ from item~(i), this allows to control $\nabla H_\delta$ if $\Gamma_\delta$ does not have clusters of 
\emph{not} $\exp(-o_{\delta\to 0}(1)\delta^{-1})$-fat faces that percolate to the boundary. 

\begin{thmas}
  \label{thmas:C1_convergence} In the setup of Theorem~\ref{thma:harmonic_functions_convergence}, assume additionally that harmonic embeddings~$\Gamma_\delta$ satisfy Assumption~\ref{assum:Exp-Fat}. %
  There exists a constant~$\beta>0$ that depends only on the constant~$\lambda$ in~\eqref{eq:Conv} such that the following holds:  
if we extend discrete harmonic functions~$H_\delta$ linearly to $\Omega\Subset U$ using the triangulations of faces from Assumption~\ref{assum:Exp-Fat}, then $h\in \mC^{1,\beta}(\Omega)$ and the gradients $\nabla H_\delta$ converge to $\nabla h$ uniformly on compact subsets of~$\Omega$. Moreover, for each compact $K\subset \Omega$ there exists a constant $C=C(K,\lambda)>0$ such that $\|h\|_{\mC^{1,\beta}(K)}\leq C\|h\|_{\mC(\overline{\Omega})}$.

Similarly, if $\Gamma_\delta$ satisfy Assumption~\ref{assum:Exp-Fat} in the setup of Theorem~\ref{thma:Green_fct_convergence}, then the gradients of~$G_{\Omega_\delta}$ converge to the gradients of $G_\Omega$ uniformly on compact subsets of $(\Omega\times\Omega)\smallsetminus\mathrm{diag}$.
\end{thmas}
\begin{rem} Let us emphasize that in Theorem~\ref{thmas:C1_convergence} we do \emph{not} impose any additional regularity assumption on the potential~$\vphi\in\mC^{1,1}$ besides the uniform convexity~\eqref{eq:Conv}. In particular, one can use Theorem~\ref{thmas:C1_convergence} to prove that $h\in \mC^{1,\beta}_\loc$ for each (weak) solution $h\in W^{1,2}_\loc$ of the equation~$\Ll_\vphi h=0$ with uniformly convex potential~$\varphi$ by discretizing the operator~$\Ll_\vphi$; see~\cite{Gutierrez-Truyen-gradient2011} for related results.
\end{rem}
\begin{proof}
  Consider the setup of Theorem~\ref{thma:harmonic_functions_convergence} first. The family $(H_\delta)_{\delta>0}$ of solutions of discrete Dirichlet problems in $\Omega_\delta$ is uniformly bounded in $\overline{\Omega}$ by the maximum principle. Combining Assumption~\ref{assum:Exp-Fat} with the alternative provided by Proposition~~\ref{prop:Hw_max_Holder_alternative}, item~(iii), and the maximum principle from Proposition~\ref{prop:Hw_max_Holder_alternative}, item~(i), we conclude that the gradients $\nabla H_\delta$ are uniformly bounded on compact subsets of $\Omega$. Moreover, the H\"older-type estimate from the item~(ii) of the same proposition and a standard pre-compactness argument implies that any subsequential limit $h$ of $H_\delta$ belongs to $\mC^{1,\beta}(\Omega)$ and that $\nabla H_\delta$ converge to $\nabla h$ uniformly on compacts along the corresponding subsequence. Since all estimates in Proposition~\ref{prop:Hw_max_Holder_alternative} are uniform in $\delta$, we also get the desired estimate $\|h\|_{\mC^{1,\beta}(K)}\leq C\|h\|_{\mC(\Omega)}$. 
  
  It remains to prove that $h$ is the (weak) solution of the Dirichlet problem for $\Ll_\vphi$ even if $\varphi\not\in\mC^{3}$. To this end we repeat the proof of Theorem~\ref{thma:harmonic_functions_convergence} given in Section~\ref{subsec:proof of Theorem12 vphi smooth} and note that we do not need to use Lemma~\ref{lemma:very_weak=weak} anymore since under Assumption~\ref{assum:Exp-Fat} we already know that each subsequential limit~$h$ of functions~$H_\delta$ is continuously differentiable inside~$\Omega$.
  
  Similar arguments apply in the setup of Theorem~\ref{thma:Green_fct_convergence}: discrete Green's functions~$G_{\Omega_\delta}$ are uniformly bounded on compact subsets of $(\Omega\times\Omega)\smallsetminus \mathrm{diag}$ due to Lemma~\ref{lemma:Green_function_bounded} which relies upon property~\ref{prty:CONV} only and does not require any additional regularity of~$\varphi$. As above, this implies that the gradients of discrete Green's functions are uniformly bounded on compacts. (Note that in Assumption~\ref{assum:Exp-Fat} we insist that each point of $\Omega$ is surrounded by a \emph{small} contour consisting of~$\delta\exp(-\delta'\delta^{-1})$-fat faces.) The proof finishes as before, with no need in Lemma~\ref{lemma:very_weak=weak} as we now know that subsequential limits of $G_{\Omega_\delta}$ are continuously differentiable.
\end{proof}

\begin{rem} \label{rem:C1-harm}
Note that in the setup of Theorem~\ref{thma:harmonic_functions_convergence} we only used that each compact set~$K\subset \Omega$ is separated from $\partial\Omega$ by a cycle of $\delta\exp(-\delta'(\delta)\delta^{-1})$-fat faces of~$\Gamma_\delta$ for all small enough~$\delta\le \delta_0(K)$. This weaker form of Assumption~\ref{assum:Exp-Fat} is not enough in the setup of Theorem~\ref{thma:Green_fct_convergence} as we also need the existence of a small circuit of $\delta\exp(-\delta'(\delta)\delta^{-1})$-fat faces separating $K$ from the pole of~$G_{\Omega_\delta}$.
\end{rem}

\subsection{Proof of Theorem~\ref{thma:harmonic_functions_convergence} and Theorem~\ref{thma:Green_fct_convergence} for general uniformly convex potentials $\vphi\in\mC^{1,1}$}
\label{subsec:Extension to C1,1}
Recall that the only place in the proofs of Theorem~\ref{thma:harmonic_functions_convergence} and Theorem~\ref{thma:Green_fct_convergence} in Section~\ref{subsec:proof of Theorem12 vphi smooth} where we used the additional regularity assumption $\vphi\in \mC^3$ is Lemma~\ref{lemma:very_weak=weak}, which is needed to prove that subsequential limits of discrete harmonic functions on $\Gamma_\delta$ are differentiable. 
In this section we provide an alternative proof of the differentiability of all such limits -- see Proposition~\ref{prop:h_W12} below -- which can be used instead of Lemma~\ref{lemma:very_weak=weak} in the general case. 

Since the statement is local, we can assume that we are given a sequence of uniformly bounded discrete harmonic functions in a disc $B(w_0,R)\Subset U$ such that $H_\delta\to h\in \mC({B(w_0,R)})$ as~$\delta\to 0$, uniformly on compacts. Our goal is to prove that $h$ is continuously differentiable and, moreover, that
\begin{equation}
\label{eq:h-in-C1}
\|h\|_{\mC^{1,\beta}({B(w_0,r)})}\ \le\ C(\lambda,r/R)\cdot \|h\|_{C({B(w_0,R)})}\ \ \text{for all}\ r<R,
\end{equation}
where the exponent~$\beta>0$ depends only on the constant~$\lambda$ in~\eqref{eq:Conv}. Let $\wtd{r}=\frac12(r\!+\!R)$.
The idea of the proof is to modify~$\Gamma_\delta$ outside the disc $B(w_0,\wtd{r})$ and to construct approximations $\wtd{H}_\delta^\eps$ of given functions $H_\delta$ such that the gradients $\nabla \wtd{H}_\delta^\eps$ are uniformly bounded in $B(w_0,\wtd{r})$ due to the alternative from  Proposition~\ref{prop:Hw_max_Holder_alternative}, item~(iii), applied on the \emph{modified} graph~$\wtd{\Gamma}_\delta$;
see also Remark~\ref{rem:C1-harm} above.

Let
\[
\wtd{V}_\delta=\left\{v_\delta\in\CC:\ \begin{array}{l}\text{either $v_\delta\in \Gamma_\delta$ and $|v_\delta-w_0|<\wtd{r}$}\\
\text{or $v_\delta\in\delta\ZZ^2$ and $\wtd{r}\le|v_\delta-w_0|\le R$}\end{array}\right\}
\]
and denote by~$\wtd{P}_\delta$ the upper convex hull of the set~$\{(v_\delta;\Phi_\delta(v_\delta)),v_\delta\in \wtd V_\delta\}\subset \CC\times\RR$. 
Since~$\Phi_\delta$ is convex, $\wtd{P}_\delta$ is the supergraph of a convex function~$\wtd{\Phi}_\delta$ that is defined on the convex hull of~$\wtd V_\delta$ and is linear on faces of a certain graph whose vertices belong to~$\wtd V_\delta$. We denote this graph by~$\wtd{\Gamma}_\delta$ and define the weights on edges of~$\wtd{\Gamma}_\delta$ using the Maxwell--Cremona correspondence.
(In other words, the weighted graph~$\wtd{\Gamma}_\delta$ and the modified potential~$\wtd{\Phi}_\delta$ are related to each other in the usual way; e.g., see~\eqref{eq:intro_H*def}.) Note that these weights are positive since $\wtd\Phi_\delta$ is convex.
It follows from~\eqref{eq:Conv} that for each small enough $\eps>0$ there exists~$\delta_0(\eps)>0$ such that for all~$\delta\le\delta_0(\eps)$
\begin{itemize}
\item all edges of~$\wtd{\Gamma}_\delta$ in the disc~$B(w_0,R-\eps)$ have length at most~$C\delta$,
\item all faces of~$\wtd{\Gamma}_\delta$ in the annulus~$B(w_0,R-\eps)\smallsetminus B(w_0,\wtd{r}+\eps)$ are $c\delta$-fat,
\item $\wtd{\Phi}_\delta$ equals $\Phi_\delta$ (and hence $\wtd{\Gamma}_\delta$ coincides with $\Gamma_\delta$ as weighted graphs) in the disc $B(w_0,\wtd{r}-\eps)$,
\end{itemize}
where the constants~$C,c>0$ depend only on the constants in property~\ref{prty:CONV}. We are now in the position to establish the smoothness of all subsequential limits of discrete harmonic functions $H_\delta$.
\begin{prop}
  \label{prop:h_W12}
  There exist an exponent $\beta>0$ depending only on the constant~$\lambda$ in~\eqref{eq:Conv} such that the estimate~\eqref{eq:h-in-C1} holds for each (subsequential, uniform on compacts) limit $h\in \mC(B(w_0,r))$ of discrete harmonic functions $H_\delta$ on $\Gamma_\delta$. 
\end{prop}
\begin{proof} Without loss of generality, assume that $\|h\|_{\mC(B(w_0,R))}=1$ and recall that the functions $H_\delta$ satisfy a H\"older-type regularity estimate from Lemma~\ref{lemma:Holder}, which is uniform in~$\delta$. Passing to the limit~$\delta\to 0$, this gives an estimate $\omega_h(\rho)\le C \rho^\beta/(R\!-\!\wtd{R})^\beta$ for the modulus of continuity of the function~$h$ in each smaller disc $B(w_0,\wtd{R})$, where $C,\beta>0$ depend on constants in~\eqref{eq:Conv} only.

Let $H_\delta^\eps$ be the solution of the discrete Dirichlet problem on~$\Gamma_\delta$ in the disc~$B(w_0,\wtd{r}+2\varepsilon)$ with boundary values $h$. Similarly, let $\wtd{H}_\delta^\eps$ be the solution of the same boundary value problem on the \emph{modified} graph~$\wtd{\Gamma}_\delta$. Recall that these two graphs coincides in~$B(w_0,\wtd{r}-\eps)$ if $\delta\le\delta_0(\eps)$. We claim that 
\begin{equation}
\label{eq:wtdH-H}
\max\nolimits_{B(w_0,\wtd{r}-\eps)}|\wtd{H}_\delta^\eps-H_\delta^\eps|\ \le\ C(\eps/\sqrt{\eps})^{\beta}+\omega_h(\sqrt{\eps})\ \ \text{for all}\ \delta\le\delta_0(\eps).
\end{equation}
Due to the maximum principle, it is enough to prove this estimate for vertices~$v_\delta$ near~$\partial B(w_0,\wtd{r}-\eps)$. Each of the two values $\wtd{H}_\delta^\eps(v_\delta)$ and $H_\delta^\eps(v_\delta)$ can be written as the expectation of the boundary value at the exit point of the corresponding random walk from $B(w_0,\wtd{r}+2\eps)$. Since both random walks are uniformly elliptic (see Lemma~\ref{lemma:crossing_estimates}), the probability that at least one of these hitting points does not belong to a $\sqrt{\eps}$-vicinity of $v_\delta$ is bounded by $O((\varepsilon/\sqrt{\varepsilon})^{\beta})$, which implies the desired upper bound. 

Since~$H_\delta\to h$ as~$\delta\to 0$ uniformly on compacts, the maximum principle implies that $H_\delta^\eps\to h$ in $B(w_0,\wtd{r})$, uniformly in~$\varepsilon$. Choose $\eps=\eps(\delta)\to 0$ so that $\delta\le\delta_0(\varepsilon)$. It follows from~\eqref{eq:wtdH-H} that 
\[
\wtd{H}_\delta^{\eps(\delta)}\to h\ \ \text{as $\delta\to 0$ uniformly on compact subsets of $B(w_0,\wtd{r})$}.
\] 
By construction, the modified graph~$\wtd{\Gamma}_\delta$ contains a circuit of $c\delta$-fat faces near the boundary of the disc $B(w_0,\wtd{r}+2\eps)$. Using Proposition~\ref{prop:Hw_max_Holder_alternative} similarly to the proof of Theorem~\ref{thmas:C1_convergence}, we conclude that the gradients $\nabla \wtd{H}_\delta^\eps$ are uniformly bounded in the disc $B(w_0,\wtd{r})$. Passing to the limit~$\delta\to 0$ in the H\"older-type estimate for the gradients $\nabla\wtd{H}_\delta^\eps$ provided by Proposition~\ref{prop:Hw_max_Holder_alternative}, item~(ii), gives~\eqref{eq:h-in-C1}. 
\end{proof}

\begin{proof}[Proof of Theorem~\ref{thma:harmonic_functions_convergence} and Theorem~\ref{thma:Green_fct_convergence} for $\vphi\in\mC^{1,1}$] As already mentioned in the beginning of this section, to prove Theorem~\ref{thma:harmonic_functions_convergence} and Theorem~\ref{thma:Green_fct_convergence} in the general case one simply repeats the proof given in Section~\ref{subsec:proof of Theorem12 vphi smooth} with Lemma~\ref{lemma:very_weak=weak} replaced by Proposition~\ref{prop:h_W12}.\end{proof}

\begin{rem}
  \label{rem:Hdelta_is_mesodiff}
  The approximation trick that we used in the proof of Proposition~\ref{prop:h_W12} does not allow us to conclude that the gradients $\nabla H_\delta$ are always uniformly bounded. However, it allows one to prove that $H_\delta$ are `differentiable at a mesoscopic scale'; the result recently obtained in~\cite{pechersky-rate2025} for harmonic functions on orthodiagonal tilings. To this end, note that in the arguments given above one can set $\wtd{r}=R-\delta^{\alpha}$, $\alpha<1$, and $\eps=C_0\delta$ with a big enough constant~$C_0$ depending on constants in property~\ref{prty:CONV} only. If one replaces~$\sqrt{\eps}$ in the proof of~$\eqref{eq:wtdH-H}$ with $\delta^{\frac12(1+\alpha)}$, the right-hand side of this estimate reads as $O(\delta^{\frac12\beta(1-\alpha)})$. Thus,
  \begin{align*}
  H_\delta(w_2)- H_\delta(w_1)&=H_\delta^\eps(w_2)- H_\delta^\eps(w_1)+O(\delta^{\frac12\beta(1-\alpha)})\\
  &=\Re\big[\,\overline{\partial_wH_\delta^\eps(w_1)}(w_2-w_1)\big]+O(|w_2-w_1|^{1+\beta})
  \end{align*}
  if $|w_2-w_1|\ge \delta^{\frac12\beta(1-\alpha)/(1+\beta)}$ and $w_1,w_2\in B(w_0,r)$, where the constant in $O(...)$ depends on~$r/R$.
\end{rem}

\begin{rem}  \label{rem:Hdelta_is_diff}
  We do not know whether or not the gradients of uniformly bounded discrete harmonic functions are always uniformly bounded if no additional regularity assumption like~\ref{assum:Exp-Fat} is imposed on~$\Gamma_\delta$ and leave this as an open question.
\end{rem}

\section{Equivalence of properties~\ref{prty:CONV}/\ref{prty:LIP} and~\ref{prty:RW}}
\label{sec:LIP<=>RW}

We begin this Section with an (elementary) proof of the equivalence of properties~\ref{prty:CONV} and~\ref{prty:LIP}. Recall that property~\ref{prty:LIP} is also equivalent to a similar Lipschitz property $\Lip(\kappa,\delta)$ of the t-surfaces~$\Theta_\delta$; see Definition~\ref{defin:Lip} and Lemma~\ref{lemma:Lip_follows_from_our_assumptions}. 

The proof of property~\ref{prty:RW} on t-surfaces satisfying $\Lip(\kappa,\delta)$ is contained in Proposition~\ref{prop:RW_ellipticity} and Lemma~\ref{lemma:white_area_is_bdd}. The main content of this section is the converse statement: the fact that~\ref{prty:RW} implies property $\Lip(\kappa,\delta)$ of the underlying t-surfaces.
\begin{lemma}
\label{lemma:Lip_equiv_conv}
(i) Assume that the potential~$\vPhi_\delta$ satisfies property~\ref{prty:CONV}
with constant~$\lambda>0$ on an open set~$U\subset\CC$. Then its gradient~$\vPsi_\delta=2\partial_w\vPhi_\delta$ satisfies property~\ref{prty:LIP} with constant~$\varkappa=\frac18\lambda$ on each open set $U'\subset U$ such that $\dist(U',\partial U)\ge 4C\delta$, where $C$ is the constant from~\eqref{eq:Conv}.\\[2pt]
(ii) Vice versa, if the map~$\vPsi_\delta=2\partial_w\vPhi_\delta$ satisfies property~\ref{prty:LIP} with constant~$\varkappa>0$ on a set $U\subset\CC$, then the potential~$\vPhi_\delta$ satisfies property~\ref{prty:CONV} with constant~$\lambda=\frac14\varkappa$ on the same set~$U$.
\end{lemma}
\begin{proof} We start with item (ii). By linearity, it is enough to check~\eqref{eq:Conv} for segments $[w_1;w_2]\subset U'$ of length $4r=|w_2-w_1|\in [2C\delta,4C\delta]$, where~$C$ is the constant from~\eqref{eq:Lip}. 
Rotating and translating the graph~$\Gamma_\delta$, we can assume that~$w_1=-2r$ and~$w_2=2r$. We have
\begin{equation}
\label{eq:x-conv-lip}
\Phi_\delta(2r)-2\Phi_\delta(0)+\Phi_\delta(-2r)=\int_{-r}^{r}\Re[\Psi_\delta(t+r)-\Psi_\delta(t-r)]dt.
\end{equation}
Due to~\eqref{eq:Lip}, we have $2\varkappa r\le \Re[\Psi_\delta(t+r)-\Psi_\delta(t-r)]\le 2\varkappa^{-1}r$ for all~$t$. This gives~\eqref{eq:Conv} with~$\lambda=\frac14\varkappa$.

Let us now pass to item~(i). Again, it is sufficient to check~\eqref{eq:Lip} for segments $[w_1;w_2]\subset U'$ of length $2r=|w_2-w_1|\in [C\delta,2C\delta]$,  where~$C$ is the constant from~\eqref{eq:Conv}. Rotating and translating the graph~$\Gamma_\delta$, we can assume that $w_1=-r$, $w_2=r$ and, moreover, that the set~$U$ contains the twice bigger disc~$\overline B(0,2r)$. Since~$\Psi_\delta$ is convex, the function~$\Re\Phi_\delta(x)$ is an increasing function of $x\in \RR$ and the identity~\eqref{eq:x-conv-lip} gives
\begin{align*}
\lambda r^2\ \le\ \Phi_\delta(r)-2\Phi_\delta(0)+\Phi_\delta(-r)\ &\le\ r\cdot \Re[\Psi_\delta(r)-\Psi_\delta(-r)]\\
&\le\ \Phi_\delta(2r)-2\Phi_\delta(0)+\Phi_\delta(-2r)\,\le\,16\lambda^{-1}r^2.
\end{align*}
To conclude the proof of~\eqref{eq:Conv}, it remains to prove a similar \emph{upper} bound for $|\Im[\Psi_\delta(r)-\Psi_\delta(-r)]|$. To this end, note that the convexity of~$\Phi_\delta$ implies that
\begin{align*}
r^{-1}(\Phi_\delta(2ir)-\Phi_\delta(r))\ &\ge\ -\Re\Psi_\delta(r)+2\Im\Psi_\delta(r),\\
r^{-1}(\Phi_\delta(-2ir)-\Phi_\delta(-r))\ &\ge\ \Re\Psi_\delta(r)-2\Im\Psi_\delta(r).
\end{align*}
Therefore, 
\[
2\Im[\Psi_\delta(r)-\Psi_\delta(-r)]\,\le\,\Re[\Psi_\delta(r)-\Psi_\delta(-r)]+r^{-1}(\Phi_\delta(2ir)-\Phi_\delta(0)+\Phi_\delta(-2ir))\,\le\,32\lambda^{-1}r.
\]
The upper bound for~$2\Im[\Psi_\delta(-r)-\Psi_\delta(r)]$ is similar and we obtain~\eqref{eq:Lip} with~$\varkappa=\frac18\lambda$.
\end{proof}

We now pass to proving that property~\ref{prty:RW} of random walks on an open set~$U$ implies property~$\Lip(\kappa,\delta)$ of the corresponding t-surfaces on each subset~$U'\Subset U$. In fact, it is sufficient to assume that $\dist(U',\partial U)\ge C\delta$ with a large enough constant~$C>0$ that depends on constants in property~\ref{prty:RW} only. We start with the following observation:
\begin{lemma}
  \label{lemma:harmonic_conjugate_bound_again}
  The assertion of Lemma~\ref{lemma:harmonic_conjugate_bound} is still valid if instead of property~\ref{prty:LIP} we assume that the graphs $\Gamma_\delta$ satisfy the upper bound on the invariant measure $\mu_\delta$ from property~\ref{assum:RW}.
\end{lemma}
\begin{proof}
  By examining the proof of Lemma~\ref{lemma:harmonic_conjugate_bound} one can see that the only way in which property~\ref{prty:LIP} is used there is the estimate of the Dirichlet energy of a smooth function using the upper bound from Lemma~\ref{lemma:white_area_is_bdd}. The latter upper bound is nothing but the estimate of the invariant measure~$\mu_\delta$ from property~\ref{prty:RW}. Thus, we can replace Lemma~\ref{lemma:white_area_is_bdd} by property~\ref{prty:RW} keeping the proof unchanged.
\end{proof}

An immediate corollary of this observation is

\begin{lemma}
  \label{lemma:vPsi_is_Lipshitz}
  Assume that $\Gamma_\delta$ have property~\ref{prty:RW} on an open set~$U\subset\CC$. There exists a constant $C>0$ depending on the constants in~\ref{prty:RW} only such that one has
  \[
    |\vPsi_\delta(w_1) - \vPsi_\delta(w_2)|\,\leq\,C|w_1 - w_2|\ \ \text{if}\ \ |w_1-w_2|\ge C\delta
  \]
  for each segment $[w_1;w_2]$ lying in the~$C\delta$-interior of~$U$.
\end{lemma}

\begin{proof}
  Due to linearity, it is sufficient to consider the case~$C_0\delta \leq |w_1-w_2| \leq 2C_0\delta$ for a large constant $C_0>0$. Let~$B=B(w_1,2|w_1-w_2|)$ and $\Hh_\delta: \Gamma_\delta\to \CC$ be the map that sends each vertex to its position in the complex plane. Provided that $C_0$ is big enough, it follows from property~\ref{prty:RW} that
  \[ %
    \Ee_U^\delta(\Hh_\delta)\,\leq\,\frac{1}{2}\sum_{v\in B_\delta} \mu_\delta(v)\,\leq\,2c^{-1}|w_1-w_2|^2.
  \] %
  Applying Lemma~\ref{lemma:harmonic_conjugate_bound} and Lemma~\ref{lemma:harmonic_conjugate_bound_again} to the real and the imaginary parts of~$\Hh_\delta$ we conclude that 
  $|\vPsi_\delta(w_1) - \vPsi_\delta(w_2)| \leq C|w_1-w_2|$, where~$C$ depends only on constants in property~\ref{prty:RW}.
\end{proof}

Recall the notation~$\Ww_\delta = \Tt_\delta - \overline \Oo_\delta$ introduced in Section~\ref{sec:T-embedding of corner graph} and denote $\Ww_\delta^\ast = \Tt_\delta + \overline \Oo_\delta$. It follows from Lemma~\ref{lemma:T-barO_projection}) that $\Ww_\delta$ projects $\Theta_\delta$ onto the harmonic embedding~$\Gamma_\delta$ while $\Ww_\delta^\ast$ projects $\Theta_\delta$ onto $\Gamma_\delta^\ast$ embedded via the conjugate harmonic embedding. 

\begin{cor}
  \label{cor:Tt(w)_is_Lipshitz}
  Assume that $\Gamma_\delta$ have property~\ref{prty:RW} on an open set~$U\subset\CC$. There exists a constant $C>0$ depending on the constants in~\ref{prty:RW} only such that one has
  \[
    |\Tt_\delta(p_1) - \Tt_\delta(p_2)|\,\leq\, C|\Ww_\delta(p_1) - \Ww_\delta(p_2)|\ \ \text{if}\ \ |\Ww_\delta(p_1)-\Ww_\delta(p_2)|\ge C\delta.
  \]
  for all pairs of points $p_1,p_2\in \Theta_\delta$ such that the segment $[\Ww_\delta(p_1);\Ww_\delta(p_2)]$ lies in the $C\delta$-interior of~$U$.
\end{cor}
\begin{proof} This directly follows from Lemma~\ref{lemma:vPsi_is_Lipshitz} as~$\Tt_\delta(p)=\Ww_\delta(p)+\Ww_\delta^\ast(p)=\Ww_\delta(p)+\Psi_\delta(\Ww_\delta(p))$\,.
\end{proof}

Given an edge $v_1v_2$ of $\Gamma_\delta$ denote by $u(v_1v_2)$ the corresponding white face of the corner graph $\Vv_\delta$ (see Definition~\ref{defin:corner_graph} and the discussion below it). 
Recall that $\mu_\delta(v_1) = 4\sum_{v_2\sim v_1}\Area[\Tt_\delta(u(v_1v_2))]$. For $\alpha\in \TT$, define
\begin{align}
\mu_\delta^\alpha(v_1) &\,=\, 2\sum_{v_2\sim v_1}\Area[(\Tt_\delta - \alpha^2\Oo_\delta)(u(v_1v_2))] \notag\\
&\,=\,\sum_{v_2\sim v_1} c_{v_1v_2}(\Im[\,\overline{\alpha}\eta_{u(v_1v_2)}\,])^2 |v_1 - v_2|^2
\,=\,\sum_{v_2\sim v_1} c_{v_1v_2} (\Re[\,\overline{\alpha}(v_1-v_2)\,])^2.
\label{eq:def_of_mu_alpha}
\end{align}
Above, $\eta$ is the origami square root function defined by~\eqref{eq:def_of_eta}, the second equality follows from~\eqref{eq:dO=etadT}, and the last equality follows from the definition of $\eta$. Note that we always have $\mu_\delta^\alpha\leq \mu_\delta$. As we will see in the next proposition, the ellipticity estimate from property~\ref{prty:RW} implies that these measures are in fact uniformly comparable to each other starting from the scale $\delta$.

\begin{prop}
  \label{prop:mu_alpha_bound}
  Assume that $\Gamma_\delta$ have property~\ref{prty:RW} 
  on an open set~$U\subset\CC$. There exist constants $c,C>0$ depending on the constants in~\ref{prty:RW} only such that for each~$\alpha\in\TT$ one has
  \[
    \textstyle \sum_{v\in B(w,r)} \mu^\alpha_\delta(v) \,\geq\, cr^2.
  \]
  for each $w\in U$ and $r\geq C\delta$ such that the disc $B(w,r)$ lies in the $C\delta$-interior of~$U$.
\end{prop}
\begin{proof}
  Denote $Z = \sum_{v\in B(w,r)} \mu_\delta(v)$ and $\nu = Z^{-1}\mu_\delta\indic_{B(w,r)}$. Recall that~$Z\ge cr^2$ as we assume property~\ref{prty:RW}.
  Let~$X_t^\nu$ be the random walk on $\Gamma_\delta$ stopped upon leaving $B(w,r)$ with $X_0$ distributed according to~$\nu$. Provided that $C$ is large enough, Proposition~\ref{prop:RW_ellipticity} and a standard large deviation estimate (e.g., see~\cite[Proposition~6.1]{CLR1}) imply that for each $\alpha\in \TT$ one has
  \begin{equation}    
  \label{eq:mab0}
    \textstyle \Var(\Re[\,\overline\alpha (X_{C\delta^2}^\nu-X_0^\nu)\,])\,\geq\,\delta^2\sum_{v\in B(w,r/2)}\nu(v)\,\geq\, \tfrac14c^2\delta^2,
  \end{equation}
  where the last inequality is a part of property~\ref{prty:RW}.
  
  Now let $X_t^{v_1}$ be the random walk on $\Gamma_\delta$ started from a fixed vertex $v_1\in\Gamma_\delta$. Using the last equality in~\eqref{eq:def_of_mu_alpha} it is easy to see that
  \[%
    \frac{d}{dt}\Var(\Re[\,\overline\alpha X^{v_1}_t\,])\big|_{t = 0}\ =\ \frac{\mu_\delta^\alpha(v_1)}{\mu_\delta(v_1)}.
  \] %
  Applying this identity to $X_t^\nu$ we obtain the following estimate:
  \begin{equation}
    \label{eq:mab2}
    \Var(\Re[\,\overline\alpha (X_{C\delta^2}^\nu-X^\nu_0)\,])\,\leq\, \int_0^{C\delta^2}\!\!\sum_{v\in B(w,r)} \frac{\mu_\delta^\alpha(v)}{\mu_\delta(v)}\nu(v)\,dt\,=\,C\delta^2 Z^{-1}\!\! \sum_{v\in B(w,r)} \mu_\delta^\alpha(v),
  \end{equation}
  where the first inequality follows from the fact that the measure~$\nu$ is sub-stationary (as~$\mu$ is stationary). Combining~\eqref{eq:mab0}, \eqref{eq:mab2}, and the estimate~$Z\ge cr^2$ one sees that $\sum_{v\in B(w,r)} \mu^\alpha_\delta(v)\geq \frac14c^2C^{-1}r^2$.
\end{proof}
We are now in the position to prove the remaining part of Theorem~\ref{thma:LIP<=>RW}.
\begin{thmas} \label{thmas:RW=>LIP} Assume that~$\Gamma_\delta$ have property~\ref{prty:RW} %
on an open set~$U$ and~$U'\Subset U$. Then, for all sufficiently small~$\delta$, the t-surfaces~$\Theta_\delta$ associated with~$\Gamma_\delta$ have property~$\Lip(\kappa,\delta)$ on~$U'$, where $k<1$ depends only on constants in~\ref{prty:RW}. Moreover, it is sufficient to assume that~$U'$ lies in the~$C\delta$-interior of~$U$ with a sufficiently large constant~$C>0$.
\end{thmas}
\begin{proof} %
  Let $R_0>0$ be a sufficiently large constant that will be chosen at the end of the proof and consider a point $p\in \Theta_\delta$ such that $B(\Ww_\delta(p), 2R_0\delta)\subset U$. It follows from Lemma~\ref{lemma:T_is_locally_one-to-one} that 
  \[
  D:=B(\Tt_\delta(p), R_0\delta)\subset\Tt_\delta(B(\Ww_\delta(p), 2R_0\delta))
  \]
  and that $\Tt_\delta$ is one-to-one if restricted to the connected component of $\Tt_\delta^{-1}(D)$ that contains~$p$; note also that it contains the disc~$B(\Ww_\delta(p),2C^{-1}R_0\delta)$ due to Corollary~\ref{cor:Tt(w)_is_Lipshitz}. In what follows, we replace~$\Gamma_\delta$ by this connected component and assume that~$\Tt_\delta$ is a global bijection from~$\Theta_\delta$ onto the disc~$D$.

 Since we now assume that $\Tt_\delta$ is one-to-one, we can identify points $p\in\Theta_\delta$ with their images $z=\Tt_\delta(p)\in D$ and view all maps defined on $\Theta_\delta$ as maps on $D$ in order to simplify the notation. For example, in this new notation we have $\Tt_\delta(z) = z$ and $\Ww_\delta(z) = z - \overline{\Oo_\delta(z)}$.  
 Corollary~\ref{cor:Tt(w)_is_Lipshitz} implies that
  \begin{equation}    
  \label{eq:pot1}
    B(\Ww_\delta(z),C^{-1}r\delta)\subset \Ww_\delta(B(z,r\delta))
  \end{equation}
  whenever $r\geq C$ and $B(z,r\delta)\subset D$ (where the constant $C>0$ is taken from Corollary~\ref{cor:Tt(w)_is_Lipshitz}).

  Let $z_0 = \Tt_\delta(p)$ and fix $R\in[\frac12R_0,R_0-2C]$ and $z_1\in D$ such that $|z_1 - z_0| = R\delta$. Further, let $\eps\in[0,1]$ and $\alpha\in \TT$ be defined by the equality
  \[
    \Oo_\delta(z_1) - \Oo_\delta(z_0)\, =\, (1-\eps)\overline\alpha^2(z_1 - z_0).
  \] 
  Our goal is to prove that~$\eps$ cannot be too close to~$0$. For~$t\in[0,1]$, denote~
  \[
    z_t=(1-t)z_0+tz_1.
  \]
  Since $|d\Oo_\delta(z_t)| = |z_1-z_0|dt$ almost everywhere and 
  \[
  \int_0^1 \Re\biggl[1-\alpha^2\frac{d}{dt}\frac{\Oo_\delta(z_t)}{z_1-z_0}\biggr]dt\ =\
     \Re\biggl[1-\alpha^2\frac{\Oo_\delta(z_1)-\Oo_\delta(z_0)}{z_1-z_0}\biggr]=\eps,
  \]
  Jensen's inequality for the concave function~
  \[
    [0,2]\ni x\mapsto \sqrt{x^2+(1-(1-x)^2)}=\sqrt{2x}
  \]
  implies that  
  \[
  \biggl|\frac{(z_t- \alpha^2\Oo_\delta(z_t)) - (z_0 - \alpha^2\Oo(z_0))}{z_1-z_0}\biggr|\ \le\ \int_0^1 \biggl|1-\alpha^2\frac{d}{dt}\frac{\Oo_\delta(z_t)}{z_1-z_0}\biggr|dt\ \le\ \sqrt{2\eps}
  \]
  for all~$t\in [0,1]$. 
  
  As $z\mapsto\Oo_\delta(z)$ is a~$1$-Lipschitz function, the mapping $D\ni z\mapsto z - \alpha^2\Oo_\delta(z)$ has no overlaps; see also~\cite[proof of Proposition~4.3]{CLR1}.  
  Let $Q$ be the $2C\delta$-neighborhood of the segment $[z_0;z_1]$. It follows from the last estimate that
  \begin{equation}
    \label{eq:pot3}
    \Area[(\Tt_\delta - \alpha^2\Oo_\delta)(Q)]\ \leq\ \pi (2C\delta + \sqrt{2\eps}|z_1 - z_0|)^2\ =\ \pi(2C + \sqrt{2\eps}R)^2\delta^2.
  \end{equation}
  On the other hand, Corollary~\ref{cor:Tt(w)_is_Lipshitz} implies that
  \[
    |\Ww_\delta(z_1) - \Ww_\delta(z_0)|\geq C^{-1}|z_1 - z_0|.
  \]
  Splitting this segment into~$\lfloor|\Ww_\delta(z_1)-\Ww_\delta(z_0)|/(2C^2\delta)\rfloor$ pieces of length~$2C\delta$, using inclusion~\eqref{eq:pot1} for discs centered at the midpoints of these pieces, definition~\eqref{eq:def_of_mu_alpha} of $\mu_\delta^\alpha$, and Proposition~\ref{prop:mu_alpha_bound}, we obtain the lower bound
  \begin{equation}
    \label{eq:pot4}
    \Area[(\Tt_\delta - \alpha^2\Oo_\delta)(Q)]\ \geq\ \frac{1}{4}\sum_{v\in \Ww_\delta(Q')} \mu_\delta^\alpha(v)\ \geq\ \frac{1}{4}\left\lfloor\frac{|z_1-z_0|}{2C^2\delta}\right\rfloor\cdot c\pi\delta^2,
  \end{equation} 
  where~$Q'$ stands for the~$C\delta$-neighborhood of~$[z_0;z_1]$. Recall that~$|z_1-z_0|=R\delta$ and~$\frac12R_0\le R\le R_0$. Provided that the constant~$R_0$ is chosen large enough, combining~\eqref{eq:pot3} and~\eqref{eq:pot4} one obtains a lower bound on~$\varepsilon$ that depends only on the constants in property~\ref{prty:RW}. 
\end{proof}

\section{Harmonicity in the intrinsic metric of the surface~$\Theta$: proof of Theorem~\ref{thma:when_harmonic}}
\label{sec:when_harmonic}

We begin with a few definitions and preliminary lemmas. Recall that $\Ll_\vphi h = -\div(A_\vphi\nabla h)$ where the matrix $A_\vphi$ is given by~\eqref{eq:def_of_Ll} and 
$\Ll^\times_\vphi h = \div(\,\ast\, A^{-1}_\vphi\!\ast \nabla h)$ is the 'dual' operator that annihilates $A$-harmonic conjugate functions of solutions of the equation~$\Ll_\vphi h=0$; see~\eqref{eq:def_of_Lltimes}. In particular, it is easy to see that $\Ll_\vphi^\times \vpsi = 0$: indeed,
the function $-i\vpsi$ is an $A$-harmonic conjugate of $w$.

\begin{lemma}
  \label{lemma:Lpsi_is_gradJac}
  Let $\Jac_w(\vpsi)=|\vpsi_w|^2-|\vpsi_{\bar{w}}|^2$ denote the Jacobian of the mapping $\vpsi:U\to\CC\cong\RR^2$. The following equation holds in the weak sense:
  \[
    \Ll_\vphi \psi = -2\partial_{\bar w} \Jac_w(\vpsi).
  \]
\end{lemma}
\begin{proof} This is a straightforward computation based upon the identity 
\[
\int_U h\Ll_\vphi g\, dxdy\ =\ 2\int_U \begin{pmatrix} h_w \!&\!h_{\bar{w}}\end{pmatrix}\begin{pmatrix} -\psi_{\bar{w}} \!&\! \overline{\psi}_{\bar w}\\ \psi_w \!&\! -\overline{\psi}_w\end{pmatrix}\begin{pmatrix} g_w \\ g_{\bar{w}}\end{pmatrix} dxdy
\]
applied to~$h=\psi$ and a smooth compactly supported test function~$g\in\mC^\infty_0(U)$.
\end{proof}

Consider now the surface $\Theta=\{(\frac12(w+\psi(w));\frac12(\overline{\psi(w)}-\overline{w}))\mid w\in U\}$  from Theorem~\ref{thma:when_harmonic}. In general, $\Theta$ is not smooth because $\psi$ is not. However $\Theta$ has tangent space almost everywhere by Rademacher's theorem since $\psi$ is Lipschitz. The inner product on $\CC^{1,1}$ induces inner products in these tangent spaces. By a slight abuse of the terminology we regard this family of inner products as a Riemannian metric on $\Theta$. %
It is easy to see that this metric is positively definite as the uniform convexity of the potential~$\vphi$ implies that $|\vpsi_{\bar{w}}|=2|\vphi_{ww}|<2\vphi_{w\bar{w}}=\vpsi_w$; see also~\eqref{eq:vpsi_ellip}.

As we endowed $\Theta$ with a Riemannian metric, we can consider its conformal parametrization. More precisely, given a domain $D\subset \CC$ we say that a map $D\ni \zeta\mapsto (z(\zeta), \theta(\zeta))\in \Theta$ is a conformal parametrization of~$\Theta$ if the mapping $\zeta\mapsto w(\zeta) = z(\zeta) + \theta(\zeta)$ is quasiconformal and we have
\begin{equation}
  \label{eq:conf_param}
  z_\zeta \bar z_\zeta = \theta_\zeta\bar\theta_\zeta
\end{equation}
almost everywhere. (This definition makes sense since quasiconformal mappings are almost everywhere differentiable and the mappings $w\mapsto z$ and $w\mapsto \theta$ are Lipschitz.) The condition~\eqref{eq:conf_param} is equivalent to requiring that the first fundamental form on $\Theta$ is (almost everywhere) diagonal in the coordinate $\zeta$. 

\begin{lemma}
  \label{lemma:metric_on_Theta}
  In the coordinate $w=x+iy$ the first fundamental form on $\Theta$ coincides with the Hessian of the potential $\vphi$. In other words, the first fundamental form on~$\Theta$ is given by
  \[
    \vphi_{xx}\,dx^2 + \vphi_{yy}\,dy^2 + 2\vphi_{xy}\,dxdy\ =\ \det A_\vphi\cdot \begin{pmatrix} dx \!&\! dy \end{pmatrix} A_\vphi^{-1} \begin{pmatrix} dx \\ dy \end{pmatrix}.
  \]
\end{lemma}
\begin{proof} A straightforward computation expresses the first fundamental form on~$\Theta$ as 
\[%
  \tfrac14\big(|\vpsi_w+1|^2 |dw|^2 +2\Re[(\vpsi_w+1)\overline{\vpsi}{}_w(dw)^2]\big)-\tfrac14\big(|\vpsi_w-1|^2|dw|^2+2\Re[(\vpsi_w-1)\overline{\vpsi}{}_w](dw)^2\big)
\]%
which equals~$\Re[\vpsi_w\,|dw|^2+\overline{\vpsi}{}_w(dw)^2]$. Note that~$\vpsi_w=\frac12(\vphi_{xx}+\vphi_{yy})$ and $\overline{\vpsi}{}_w=\frac12(\vphi_{xx}-\vphi_{yy})-i\vphi_{xy}$.
\end{proof}

\begin{cor}
  \label{cor:conf_param_diag}
  Denote~$\rho=(\det A_\vphi)^{1/2}=(\det D^2\vphi)^{1/2}$ and assume that $D\ni\zeta\mapsto (z(\zeta),\theta(\zeta))$ is a conformal parametrization of $\Theta$. Then,
  \begin{align*}
    \Jac_\zeta(w)\cdot \Ll_\vphi h\ &=\ -\div_\zeta(\rho\,\nabla_\zeta h),\\
    \Jac_\zeta(w)\cdot \Ll_\vphi^\times h\ &=\ -\div_\zeta(\rho^{-1}\nabla_\zeta h)
  \end{align*}
  for each $h$ with locally square summable gradient, where equalities are understood in the weak sense. (Note that $\rho$ is bounded away from zero and infinity since we assume that $\vphi$ is uniformly convex.)
\end{cor}
\begin{proof} 
  By definition of a conformal coordinate, the matrix~$A_\vphi^{-1}\det A_\vphi$ of the first fundamental form on~$\Theta$ becomes a scalar function in the coordinate~$\zeta$. Therefore, $A_\vphi^{-1}=\rho^{-1}(\det J)^{-1}J^\top J$, where $J\in\RR^{2\times 2}$ is the Jacobian matrix of the coordinate change~$w(\zeta)$; in particular,~$\Jac_\zeta(w)=\det J$. 
  Taking the inverse, we see that $A_\vphi=(\det J)\cdot\rho\,J^{-1}(J^{-1})^\top$ and we also have $\ast\,A_\vphi^{-1}\ast = -(\det J)\cdot\rho^{-1}J^{-1}(J^{-1})^\top$ since $\ast\,J^\top\!= (\det J)\cdot J^{-1}\,\ast$. %
  The claim now follows by examining how  quadratic forms of~$\Ll_\vphi$ and~$\Ll_\vphi^\times$ change under this coordinate change. 
\end{proof}

One says that a surface $\Theta\subset\CC^{1,1}\cong\RR^{2,2}$ is \emph{maximal} %
if it is smooth and its mean curvature vector vanishes; we emphasize that the mean curvature is calculated with respect to the Minkovski metric~\eqref{eq:intro-R2,2metric-def}. 
As in Euclidean spaces, vanishing of the mean curvature means that $\Theta$ is locally a stationary point of the area functional; however, for space-like surfaces in~$\RR^{2,2}$ such stationary points correspond to maxima rather than minima, hence the name. A direct computation similar to the classical Euclidean setup shows that $\Theta$ is maximal if and only if its conformal parametrization $(z(\zeta), \theta(\zeta))$ is harmonic, i.e.,
\begin{equation}
  \label{eq:maximal}
  z_{\zeta\bar\zeta} = \theta_{\zeta\bar\zeta} = 0.
\end{equation}
(Note that this is equivalent to saying that~$w$ and~$\vpsi$ are harmonic in~$\zeta$.)
In particular, maximal surfaces are real analytic.
We are now in the position to prove Theorem~\ref{thma:when_harmonic}.

\begin{proof}[Proof of Theorem~\ref{thma:when_harmonic}]
  We have already proved (i) in Lemma~\ref{lemma:metric_on_Theta}, let us prove (ii). 
  
  First, assume that assertion~(a) holds, i.e., that $\det D^2\vphi=\rho^2$ is constant. 
  Since~$\Ll_\vphi w=0$ and~$\Ll_\vphi^\times \psi=0$, Corollary~\ref{cor:conf_param_diag} implies that~$w$ and~$\vpsi$ are harmonic functions of~$\zeta$, that is,~\eqref{eq:maximal} holds and the surface~$\Theta$ is maximal as required in~(b).

  Second, let us prove that~(b) implies~(c). %
  Applying Corollary~\ref{cor:conf_param_diag} again and using the fact that $w$ is harmonic in $\zeta$ we can write
  \[
    0\ =\ \Jac_\zeta(w)\Ll_\vphi w\ =\ \div_\zeta(\rho\nabla_\zeta w) = \nabla_\zeta\rho\cdot\nabla_\zeta w\,,
  \]
  which is only possible if~$\nabla_\zeta\rho=0$ since~$\rho$ is a real-valued function and the map~$\zeta\mapsto w$ preserves orientation. Therefore, both $\Ll_\vphi$ and~$\Ll_\vphi^\times$ are proportional to the Laplacian in the coordinate $\zeta$.
  
  Third, let us prove that (c) implies (d). Let~$w(\zeta)$ be the coordinate change from (c); note that we do \emph{not} assume in advance that~$\zeta$ is a conformal parametrization of~$\Theta$. However, as we have $\Ll_\vphi[\zeta]=0$ and~$\Ll_\vphi[\zeta^2]=0$, the operator~$\Ll_\vphi$ has to be proportional to the Laplacian in the coordinate $\zeta$. Therefore, the Jacobi matrix of $w(\zeta)$ diagonalizes $A_\vphi$, which implies that $\zeta$ \emph{is} a conformal parametrization of~$\Theta$. Using Corollary~\ref{cor:conf_param_diag} and~$\Ll_\vphi[\zeta]=0$ again, we see that~$\nabla_\zeta\rho=0$. Hence, $\rho$ is constant and the operators~$\Ll_\vphi$ and $\Ll_\vphi^\times$ are multiples of each other.
  
  Finally, (d) trivially implies (e) since~$\Ll_\vphi\psi=0$, and it is also easy to see that (e) implies (a). Indeed, it follows from Lemma~\ref{lemma:Lpsi_is_gradJac} that $\partial_{\bar{w}}\Jac_w(\psi)=-\frac12\Ll_\vphi\vpsi=0$, which means that the real-valued function~$\Jac_w(\vpsi)=\det A_\vphi=\det D^2\vphi$ is constant.
  \end{proof}

\printbibliography

\end{document}